\RequirePackage{tikz} 

\documentclass[]{article}
\usepackage[a4paper, total={6.5in, 8in}]{geometry}

\usepackage{dsfont}
\usepackage{mathtools}
\usepackage{amsthm}
\usepackage{amsmath}
\usepackage{amssymb}
\usepackage{hyperref}

\DeclareMathOperator*{\argmin}{arg\,min}
\DeclareMathOperator{\cost}{Cost}

\DeclareMathOperator{\spl}{Split}
\DeclareMathOperator{\merge}{Merge}
\DeclareMathOperator{\move}{Move}

\DeclareMathOperator{\val}{val}

\usepackage[backend=biber,
style=numeric,
sorting=ynt
]{biblatex}
\newtheorem{lemma}{Lemma}
\newtheorem{corollary}{Corollary}

\begin{document}

\providecommand{\keywords}[1]
{
  \small	
  \textbf{\textit{Keywords---}} #1
}

\title{On Computing Exact Means of Time Series Using the Move-Split-Merge Metric}
\author{Jana Holznigenkemper\footnote{ Dept. of Mathematics and Computer Science, University of Marburg, Germany, \{holznigenkemper, komusiewicz, seeger\}@mathematik.uni-marburg.de} , Christian Komusiewicz$^*$, Bernhard Seeger$^*$  \\
}
\date{} 

\maketitle


\abstract{Computing an accurate mean of a set of time series is a critical task in applications like nearest-neighbor classification and clustering of time series. 
While there are many distance functions for time series, the most popular distance function used for the computation of time series means is the non-metric dynamic time warping (\textsc{DTW}) distance.
A recent algorithm for the exact computation of a \textsc{DTW-Mean} has a running time of $\mathcal{O}(n^{2k+1}2^kk)$, where $k$ denotes the number of time series and $n$ their maximum length.
In this paper, we study the mean problem for the move-split-merge (\textsc{MSM}) metric that not only offers high practical accuracy for time series classification but also carries of the advantages of the metric properties that enable further diverse applications.
The main contribution of this paper is an exact and efficient algorithm for the \textsc{MSM-Mean} problem of time series. The running time of our algorithm is $\mathcal{O}(n^{k+3}2^k k^3 )$, and thus better than the previous \textsc{DTW}-based algorithm.
The results of an experimental comparison confirm the running time superiority of our algorithm in comparison to the \textsc{DTW-Mean} competitor.
Moreover, we introduce a heuristic to improve the running time significantly without sacrificing much accuracy.}


\section{Introduction}
\label{chap:intro}
Time series databases have gained much attention in academia and industry due to demands in many new challenging applications like Internet of Things (IoT), bioinformatics, social and system monitoring. In particular, because of the emergence of IoT, the requirement for developing dedicated systems \cite{garcia2020db2} supporting time series as a first-class citizen has increased recently. In addition to supporting fundamental database operations like filters and joins, analytical operations like clustering and classification are highly relevant in time series databases.

The analysis of time series like clustering largely depends on the underlying distance functions. In a recent study, Paparrizos et al. \cite{paparrizos2020debunking} re-examined the impact of 71 distance functions on classification for many datasets. While \emph{dynamic time warping (DTW)} and related functions \cite{berndt1994using} 
had the reputation of being the best choice, 

Paparrizos et al. \cite{paparrizos2020debunking} found DTW performing inferior for time series classification in comparison to many other elastic distance functions. Among those is the \emph{move-split-merge (\textsc{MSM}) metric} \cite{stefan2012move}. It works similarly to the Levensthein distance \cite{levenshtein1966binary} by transforming one time series into another using three types of operations. A move operation changes the value of a data point, a merge operation fuses two consecutive points with equal values into one, and a split operation splits a point into two adjacent points with the same value. In addition to its superiority to DTW, \textsc{MSM} offers another significant advantage: it satisfies the properties of a mathematical metric, and thus it is ready-to-use for metric indexing \cite{novak2011metric} and algorithms that presume the triangle inequality.

Partition-based algorithms such as k-means clustering are among the best methods for clustering time series~\cite{paparrizos2017fast}. One of the fundamental problems of k-means clustering for time series is how to compute a mean for a set of time series. Brill et al.~\cite{brill2019exact} studied the problem for DTW and developed an algorithm computing an exact mean of $k$ time series in~$\mathcal{O}(n^{2k+1}2^kk)$ time, where $n$ is the maximum length of an input time series. To the best of our knowledge, the mean problem of time series has not been addressed for other distance functions like the \textsc{MSM} metric so far. 

In this paper, we examine the mean problem of time series for the \textsc{MSM} metric. The mean $m$ of a set $X$ of input time series is a time series that minimizes the sum of the distances to the time series in $X$ regarding the \textsc{MSM} metric. In the following, we use \textsc{MSM-Mean} and \textsc{DTW-Mean} to denote the mean of the \textsc{MSM} metric and DTW distance function, respectively. 

An example of \textsc{MSM-Mean} is depicted in Figure \ref{fig:italy}. It comprises four sample time series from the data set \emph{Italy Power Demand} of the UCR time series archive \cite{UCRArchive} with their respective \textsc{MSM-Mean}. 
In contrast to \textsc{DTW-Mean}, we show that \textsc{MSM-Mean} consists of values only present in the underlying time series. 
This observation is crucial for the design and efficiency of our algorithm. 
We prove that the running time of our algorithm is $\mathcal{O}(n^{k+3}2^k k^3 )$, thus faster than the \textsc{DTW}-based competitor \cite{brill2019exact}.

In summary, our contributions are: 
\begin{itemize}
    \item We give new essential characteristics of the \textsc{MSM} metric. We first prove that there always exists an optimal transformation graph that is a forest to further specify the values of some crucial nodes within this forest.
    \item We show that \textsc{MSM-Mean} only consists of data points that are present in at least one time series of the input set. 
    \item We develop a dynamic program computing the exact \textsc{MSM-Mean} of $k$ input time series achieving a better theoretical running time than its competitor \textsc{DTW-Mean}.
    \item In experiments on samples of real-world time series, we show that the practical running time of \textsc{MSM-Mean} is faster than \textsc{DTW-Mean}.
    \item We present preliminary heuristics computing \textsc{MSM-Mean} in significantly faster running time without sacrificing much accuracy. 
\end{itemize}

The remainder of the paper is structured as follows. 
Section~\ref{chap:related_work} reviews related work. 
In Section~\ref{chap:preliminaries}, we give some important preliminaries for the \textsc{MSM} metric and formulate the \textsc{MSM-Mean} problem. 
Then, in Section~\ref{chap:properties_MSM_Metric}, we introduce some new properties of the  \textsc{MSM} metric to prove at the end of the section that there always exists a mean consisting of data points of the input time series. 
The dynamic program for the exact \textsc{MSM-Mean} algorithm is given in Section~\ref{chap:computing_MSM_mean}.
We experimentally evaluate our approach, discuss various heuristics, and compare it to the \textsc{DTW-Mean} algorithm in Section~\ref{chap:experimental_eval},
and conclude in Section~\ref{chap:conclusion}.
 
\begin{figure}[t]
    \centering
    \includegraphics[width = 0.5\textwidth]{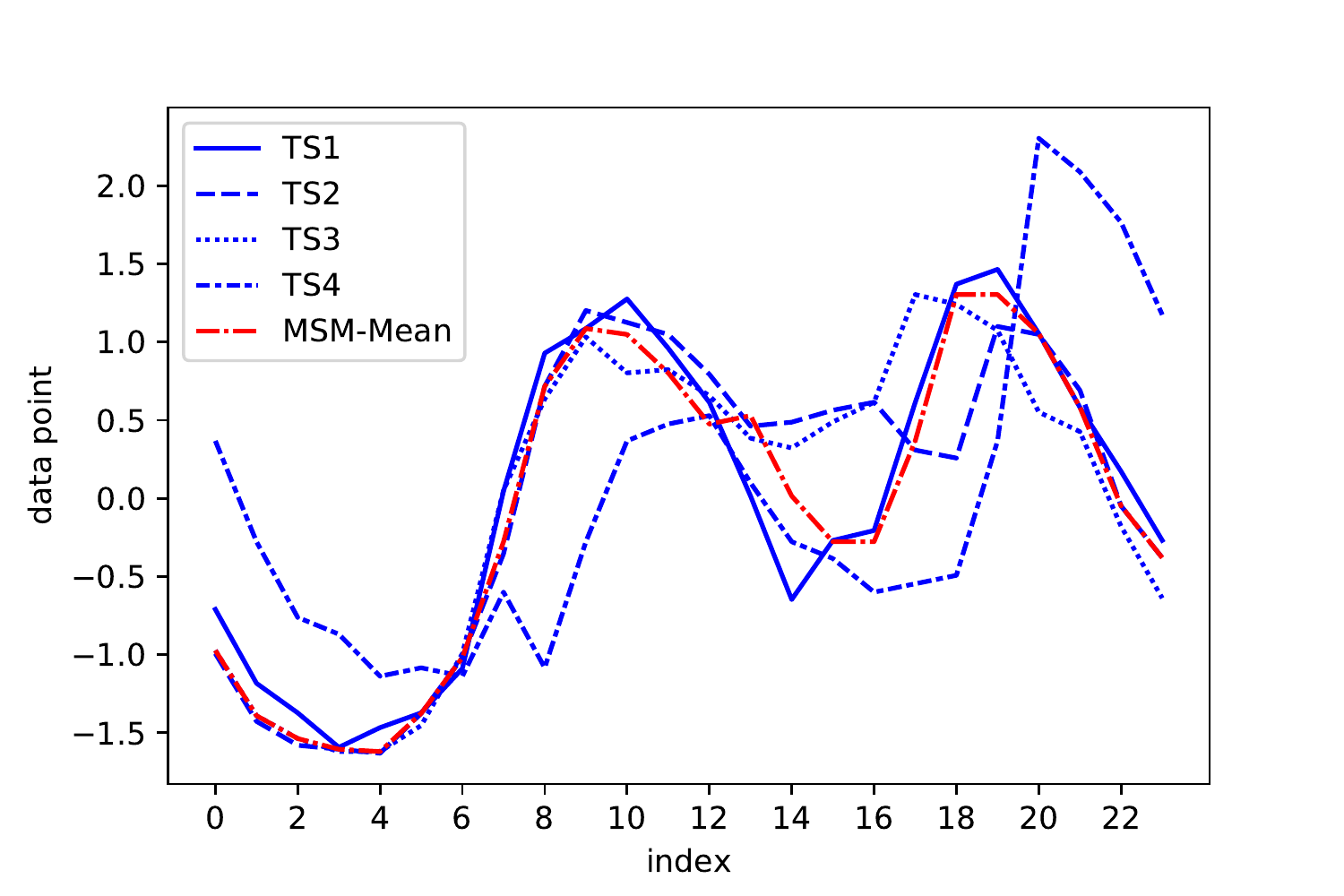}
    \caption{\textsc{MSM-Mean} of four time series from the \textit{Italy Power Demand} data set containing time series of length $n=24$.}
    \label{fig:italy}
\end{figure}

\section{Related Work}
\label{chap:related_work}
For the exploratory analysis of time series, clustering is used to discover interesting patterns in time series datasets \cite{das1998rule}. Much research has been done in this area \cite{liao2005clustering}.  The surveys \cite{aghabozorgi2015time, rani2012recent} give a recent overview of many methods. The problem of mean computation is discussed for Euclidean distance and for the DTW distance, but not for the MSM metric.

Moreover, the use of classification methods is indispensable for accurate time series analysis~\cite{jiang2020time}. 
The temporal aspect of time series has to be taken into account for clustering and classification, though finding a representation of a set of time series is a challenging task.
Determining accurate means of time series is crucial for partitioning clustering approaches \cite{niennattrakul2007clustering} like k-means \cite{macqueen1967some}, where the prototype of a cluster is a mean of its objects, and for nearest-neighbor classification  \cite{petitjean2016faster}. 
These methods are based on the choice of the underlying distance function.
Among the existing time series distance measures \cite{paparrizos2020debunking}, the DTW distance \cite{sakoe1978dynamic} is a very important measure with application in, e.g., similarity search \cite{sakurai2005ftw}, speech recognition \cite{kruskal1983time}, or gene expression \cite{aach2001aligning}.
We now give an overview of mean computation methods using the \textsc{DTW} distance. \\
Besides the exact \textsc{DTW-Mean} algorithm~\cite{brill2019exact} 
 that minimizes the Fréchet function \cite{frechet1948elements}, there are many heuristics trying to address this problem.
Some approaches first compute a multiple alignment of $k$ input time series and then average the aligned time series column-wise~\cite{hautamaki2008time,petitjean2012summarizing}.
 \emph{DTW barycenter averaging (DBA)}~\cite{petitjean2011global} is a heuristic strategy that iteratively refines an initially average sequence in order to minimize its square DTW-distance to the input time series. 
Other approaches exploit the properties of the Fréchet function \cite{cuturi2017soft,schultz2018nonsmooth}. Their methods are based on the observation that the Fréchet function is Lipschitz continuous and thus differentiable almost everywhere.  
 Cuturi et al. \cite{cuturi2017soft} use a smoothed version of the DTW-distance to obtain a differentiable Fréchet DTW-distance. 
Brill et al. \cite{brill2019exact} showed that none of the aforementioned approaches is sufficiently accurate compared to the exact method. Since clustering methods based on partitioning rely on cluster prototype determination \cite{aghabozorgi2015time}, it is necessary to compute an accurate mean. 
All these observations make it indispensable to consider the problem for other distance functions, like the \textsc{MSM} metric.

The \textsc{MSM} metric is already investigated for classification.
One of the first studies of the \textsc{MSM} metric concerning its application for classification problems was by Stefan et al. \cite{stefan2012move}. They perform their tests on 20 data sets of the UCR archive \cite{UCRArchive}. The \textsc{MSM} distance is tested against the DTW distance, the constrained DTW distance, the edit distance on real sequence and the Euclidean distance. For a majority of the tests, the \textsc{MSM} distance performs better than the compared measures. 
There have been further studies of the accuracy of different time series distance measures regarding 1-NN classification problems \cite{bagnall2017great, lines2015time}. 
Bagnall et al. \cite{bagnall2017great} also conclude that the \textsc{MSM} distance leads to better accuracy results than DTW but worse running time.   
All these studies come to a similar result as the most recent study of Paparrizos et al. \cite{paparrizos2020debunking}. 
To the best of our knowledge, there are no studies that investigate and extend the theoretical concepts and applications of the \textsc{MSM} distance. 

The subject of time series, also known as data series, has attracted attention within the database research domain recently, see \cite{jensen2017time} for a recent survey. There are time series data bases, also known as event stores, that are specially designed for the analysis of time series \cite{bader2017survey, garcia2020db2}. These systems rarely support clustering, but focus on supporting the basic building blocks for query processing. 

Since the \textsc{MSM} distance obeys all properties of a mathematical metric, especially the triangle inequality, it also applies to problems like metric indexing \cite{DBLP:journals/pvldb/ChenGZJYY17, novak2011metric}. In fact, metric indexing also requires the computation of pivots that is closely related to the mean. However, pivots belong to the underlying data set, while the mean (of a time series) is generally a newly generated object.

\section{Preliminaries}
\label{chap:preliminaries}

Let us first introduce our notation and problem definition. For~$k\in\mathds{N}$, let $[k]:= \{1,\ldots ,k\}$. 
A \emph{time series} of length $n$ is a sequence $x=(x_1, \ldots , x_n)$, where each \emph{data point}, in short \emph{point},  $x_i$ is a real number.
Let $V(x) = \{x_i~\vert~x~\in~x\}$ be the set of all values of points of $x$. 
For $i<j$, the point $x_i$ is a \emph{predecessor} of the point~$x_j$ and the point~$x_j$ is a \emph{successor} of the point~$x_i$.
For a set of time series $X = \{x^{(1)}, \ldots ,x^{(k)}\}$, the  $i$th point of the $j$th time series of~$X$ is denoted as $x_i^{(j)}$; time series~$x^{(j)}$ has length $n_j$. Further let $V(X) = \cup_{j\in[k]} V(x^{(j)}) = \{v_1,\ldots, v_r\}$ be the set of the values of all points of all time series in $X$.

\subsection{Move-Split-Merge Operations}
We now define the \textsc{MSM} metric, following the notation of Stefan et al. \cite{stefan2012move}, and the \textsc{MSM-Mean} problem. 
The \textsc{MSM} metric allows three transformation operations to transfer one time series into another: 
\emph{move, split}, and \emph{merge} operations. 
For time series $x=(x_1, \ldots , x_n)$ a move transforms a point $x_i$ into $x_i + w$ for some~$w\in\mathds{R}$, that is,  $\move_{i,w}(x)\coloneqq(x_1, \ldots , x_{i-1},x_i+w, x_{i+1}, \ldots ,x_n) $, with cost $\cost(\move_{i,w}) = \vert w \vert$.
Informally, we say that there is a \emph{move at point $x_i$ to another point} $x_i~+~w$. The split operation splits the $i$th element of $x$ into two consecutive points. 
A split at point $x_i$ is defined as $\spl_i(x) \coloneqq (x_1, \ldots , x_{i-1}, x_i, x_i, x_{i+1}, \ldots, x_n).$ 

A merge operation may be applied to two consecutive points of equal value. 
For $x_i = x_{i+1}$, it is given by $\merge_i(x)\coloneqq(x_1, \ldots , x_{i-1}, x_{i+1}, \ldots, x_n).$
We say that~$x_i$ and $x_{i+1}$ \emph{merge to a point}~$z$.
Split and merge operations are inverse operations. 
Their costs are assumed to be equal and determined by a given nonnegative constant $c  = \cost(\spl_{i}) =  \cost(\merge_{i})$. 
A \emph{sequence of transformation operations} is given by $\mathbb{S} = (S_1, \ldots , S_s)$, where $S_j \in \{\move_{i_j,w_j}, \spl_{i_j}, \merge_{i_j}\}$. 
A \emph{transformation} $T(x, \mathbb{S})$ of a time series $x$ for a given sequence of transformation operations $\mathbb{S}$ is defined as $T(x, \mathbb{S}) \coloneqq T(S_1(x), (S_2, \ldots , S_s))$.
If $\mathbb{S}$ is empty, we define $T(x, \emptyset) \coloneqq x$. The cost of a sequence of transformation operations $\mathbb{S}$ is given by the sum of all individual operations cost, that is, $ \cost(\mathbb{S}) \coloneqq \sum_{S \in \mathbb{S}}\cost(S).$
We say that $\mathbb{S}$ transforms  $x$ to $y$ if $T(x,\mathbb{S}) = y$.
We call a transformation an \emph{optimal transformation} if it has minimal cost transforming $x$ to $y$.
The \emph{\textsc{MSM} distance} $d(x,y)$ between two time series $x$ and $y$ is defined as the cost of an optimal transformation. 
The distance $D(X,y)$ of multiple time series $X=\{x^{(1)},\ldots,x^{(k)}\}$ to a time series $y$ is given by $D(X,y) = \sum_{x\in X} d(x,y). $
A \emph{mean} $m$ of a set of time series $X$ is defined as the time series with minimum distance to $X$, that is, $ m = \argmin_{z\in \mathcal{Z}} D(X,z)$,
where $\mathcal{Z}$ is the set of all finite time series. 
The problem of computing a mean is thus defined as follows:\\
\\
\textsc{MSM-Mean}\\
\textsc{Input:} A set of time series $X = \{x^{(1)}, \ldots,x^{(k)}\}$. \\
\textsc{Output:} A time series $m$ such that $ m = \argmin_{z\in \mathcal{Z}} D(X,z)$. \\

Before we regard the \textsc{MSM-Mean} problem in more detail, we will introduce the concept of \emph{transformation graphs} to describe the structure of a transformation $T(x,\mathbb{S}) = y$.

\subsection{Transformation graphs}

The transformation $T(x,\mathbb{S}) = y$ can be described by a directed acyclic graph $G_{\mathbb{S}}(x,y)$, the \emph{transformation graph},  with \emph{source nodes} $N(x) = \{u_1,\ldots,u_m\}$ and \emph{sink nodes} $N(y) = \{v_1,\ldots,v_n\}$, where a node~$u_i$ represents the point~$x_i$ and the node~$v_j$ represents the point~$y_j$.
All nodes which are neither source nor sink nodes are called \emph{intermediate nodes}.
If the time series and operation sequence are clear from context, we may write $G$ instead of $G_{\mathbb{S}}(x,y)$. 
Each node in the node set $V$ of $G$ is associated with a value given by a function $\val:~V \rightarrow \mathds{R}$. 
For source and sink nodes we have $\val(u_i) = x_i$ and $\val(v_j) = y_j$. 
Each intermediate node is also associated with a value.
The edges represent the transformation operations of $\mathbb{S}$. 
To create a transformation graph, for each operation in $\mathbb{S}$ a respective \emph{move edge} or two \emph{split}, or \emph{merge edges} are added to the graph. 
A move edge can be further specified as an \emph{increasing (inc-)} or \emph{decreasing (dec-)} edge if the move operation adds a positive or negative value to the value of the parent node, respectively. 
An edge can be either a move, split or merge edge. 
If a node $\alpha$ is connected to a node $\beta$ by a split edge and $\beta$ is a child of $\alpha$, then there exists a node $\gamma\neq \beta$ to which $\alpha$ is connected by a split edge and which is a child of $\alpha$. 
If the nodes $\alpha$ and $\beta$ are connected by a merge edge and $\alpha$ is a parent of $\beta$, then there exists a node~$\gamma \neq \alpha$ which is connected to $\beta$ by a merge edge and is a parent of $\beta$. 
Moreover, for the split and the merge case, it holds that $\val(\alpha) = \val(\beta)  = \val(\gamma)$. 

Given a set of sequence of operations $\mathbb{S}$, the transformation graph $G_{\mathbb{S}}(x,y)$ is unique.
Given a transformation graph $G$, it may be derived from different sequences of operations since a sequence $\mathbb{S}$ is only partially ordered.
Taking the example of a transformation graph (see Figure \ref{fig:example_transformation_graph}), that means, that for example the move operation between the node $u_1$ and $\alpha$ and the move operation between $u_4$ and $v_3$ are interchangeable.   

\begin{figure}[t]
    \centering
    \includegraphics[width = 0.25\textwidth]{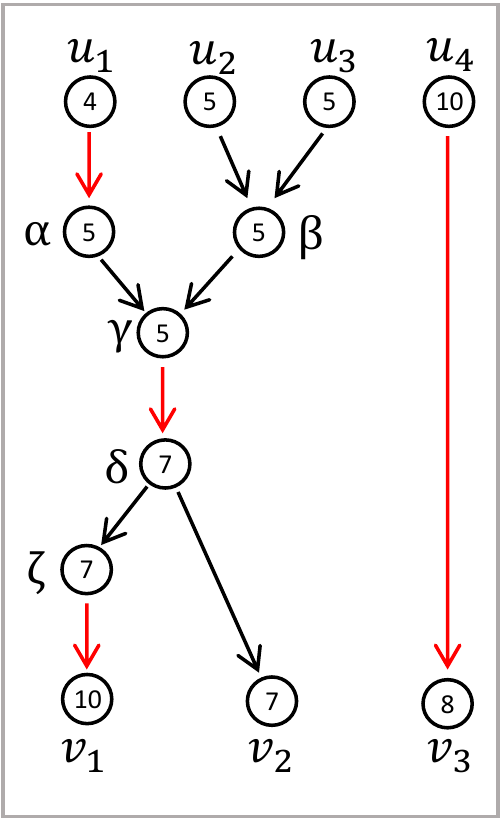}
    \caption{Optimal transformation graph of $x=(4,5,5,10)$ to $y=(10,7,8)$ for $c=0.1$. Move edges are colored in red in this work. The cost of a move edge is the difference between the source and the target point. We have total cost merge and split cost $3c$ and move cost of 8. Hence, the distance between $x$ and $y$ is $d(x,y) = 8.3$.}
    \label{fig:example_transformation_graph}
\end{figure}

A \emph{transformation path}, in short \emph{path}, in~$G_{\mathbb{S}}(x,y)$ is a directed path from a source node $u_i \in N(x)$ to a sink node $ v_j \in N(y)$. 
We say that $u_i$ is \emph{aligned} to $v_j$. 
A path can be further characterized by its sequence of edge labels. For example, in Figure \ref{fig:example_transformation_graph}, the path from $u_1$ to $v_2$ is an inc-merge-inc-split path. 
Analogously, we say that the path consists of consecutive inc-merge-inc-split edges. 

A transformation path is \emph{monotonic} if the move edges on this path are only inc- or only dec-edges.  
A monotonic path may be specified as \emph{increasing} or \emph{decreasing}. 
A transformation is \emph{monotonic} if the corresponding transformation graph only contains monotonic paths. 
A transformation graph is \emph{optimal} if it belongs to an optimal transformation.
Two transformation graphs are \emph{equivalent} if they have the same sink and source nodes.

In the next section, we recap some known properties about the transformation graph and extend them proving some new essential characteristics.

\subsection{Properties of Transformation Graphs}
In the following, we summarize some important known properties about the transformation graph by Stefan et al. \cite{stefan2012move}. 
The first lemma states that there exists an optional transformation graph without split and merge edges that occur directly after another on a path.
\begin{lemma}[Proposition 2 \cite{stefan2012move}]
\label{prop:merge_split}
For any two time series~$x$ and~$y$, there exists an optimal sequence of transformation operations~$\mathbb{S}$ such that $G_{\mathbb{S}}(x,y)$ contains no consecutive merge-split or split-merge edges. 
\end{lemma}

By construction, two consecutive move-move edges are not useful, since they can be combined into one move edge.
We extend Lemma \ref{prop:merge_split} to further path restrictions in an optimal transformation graph. 
That is, that there exists an optimal transformation graph without paths containing consecutive split-move-merge edges. 

\begin{lemma}
\label{prop:split-move-merge}
For any two time series~$x$ and~$y$, there exists an optimal sequence of transformation operations~$\mathbb{S}$ such that $G_{\mathbb{S}}(x,y)$ contains no consecutive split-move-merge edges. 
\end{lemma}

\begin{proof}
Assume an optimal transformation graph $G_{\mathbb{S}}(x,y)$ including split-move-merge edges. Since the underlying set of transformation operations $\mathbb{S}$ of $G$ is partially ordered, we can reorder the operations in $\mathbb{S}$, choosing an order, where split, move and merge operations are directly applied after one another. 
Figure \ref{fig:lemma_split_move_merge} shows two different possibilities how consecutive split-move-merge edges may be contained in a transformation graph.  \\
\indent\textit{Case I}: We consider a split at node $\alpha$ to the nodes $\alpha'$ and $\alpha''$ where $\val(\alpha) = \val(\alpha') = \val(\alpha'')$. It is followed by two move edges from $\alpha'$ to $\beta'$ and from $\alpha''$ to $\beta''$ and a merge of $\beta'$ and $\beta''$ (see Figure \ref{fig:lemma_split_move_merge}a). Therefore, the values added to $\val(\alpha)$ have to be equal on both move edges, that is a value $w\in \mathds{R}$. 
The cost of these transformation operations are $2c + 2\vert w \vert$. Consider replacing the two split-move-merge edges with one direct move edge from $\alpha$ to $\beta$ adding $w$ to $\val(\alpha)$ (see Figure \ref{fig:lemma_split_move_merge}b). This replacement leads to an equivalent transformation with cost  $\vert w \vert < 2c + 2\vert w \vert$. 
This is a contradiction to our assumption that $G$ is optimal. \\
\indent\textit{Case II}: Consider the part of a transformation graph in Figure \ref{fig:lemma_split_move_merge}c. There is a split at $\alpha$ to $\alpha'$ and $\alpha''$. The node $\alpha''$ is connected by a move edge to $\beta'$ adding a value $w$ to $\val(\alpha'')$. The node $\beta'$ merges with $\beta''$ to $\beta$. 
Deleting the split-move-merge edge and editing the part of the graph as shown in Figure \ref{fig:lemma_split_move_merge}d leads to an equivalent transformation graph, saving cost of $2c +\vert w \vert$. This is a contradiction to our assumption that $G$ is optimal. 
\end{proof}

\begin{figure}
    \centering
    \includegraphics[width = 0.9\textwidth]{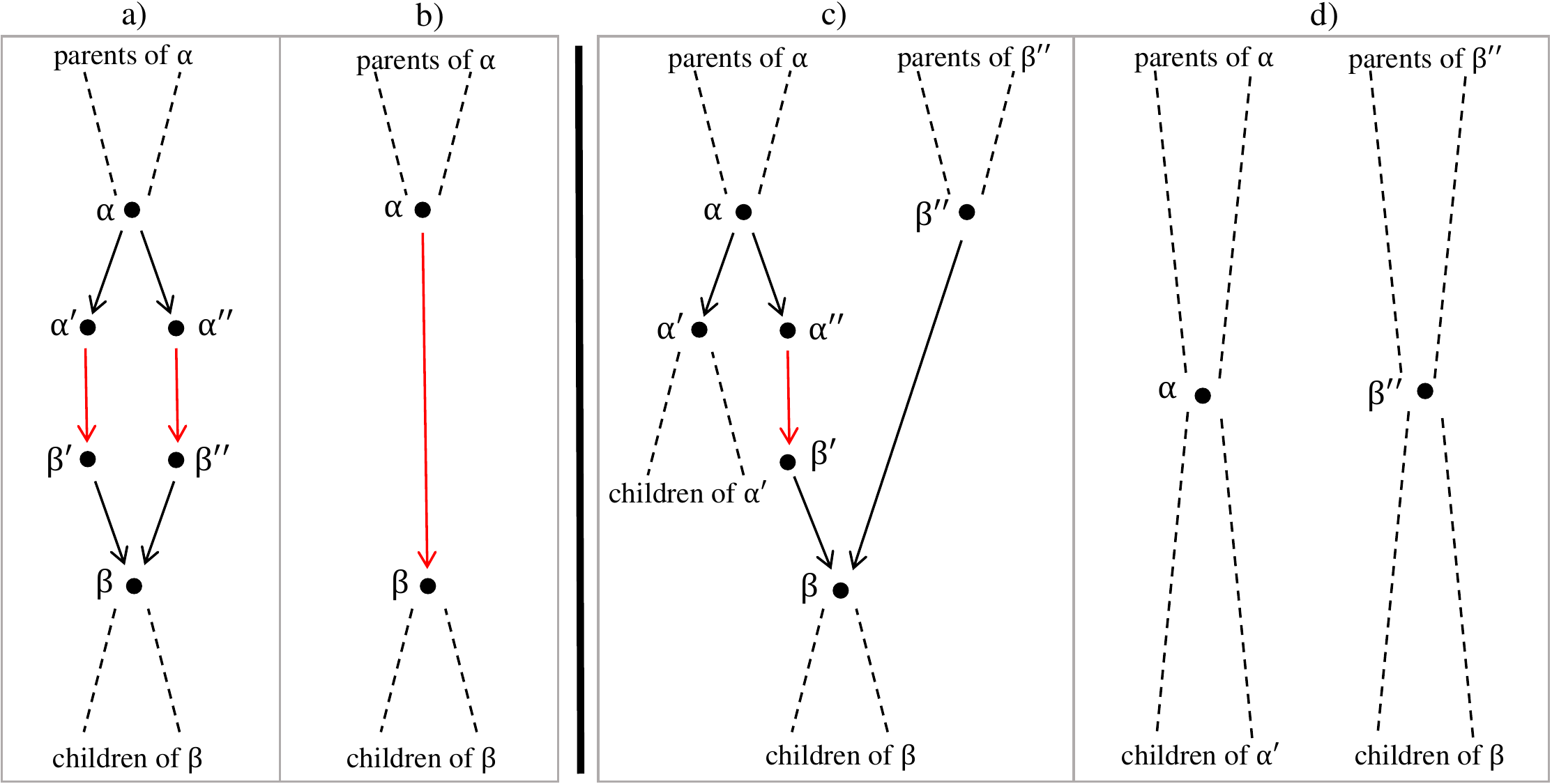}
    \caption{a) First possibility of a transformation graph including consecutive split-move-merge edges. b) Equivalent transformation graph to a). c) Second possibility of a transformation graph including consecutive split-move-merge edges. d) Equivalent transformation graph to~c)}
    \label{fig:lemma_split_move_merge}
\end{figure}

The next lemma states that there is always an optimal monotonic transformation.

\begin{lemma}[Monotonicity lemma \cite{stefan2012move}]
\label{lemma:monotonicity}
For any two time series $x$ and $y$, there exists an optimal transformation that converts $x$ into $y$ and that is monotonic. 
\end{lemma}

Summarizing the above properties, there always exists an optimal transformation graph only containing paths from source to sink nodes of the following consecutive edge types: 
\begin{description}
    \item[Type 1:] move - move - $\cdots$ - move - move 
    \item[Type 2:] split/move - split/move - $\cdots$ - split/move -split/move
    \item[Type 3:] merge/move - merge/move - $\cdots$ - merge/move - merge/move
    \item[Type 4:] Type 3 - merge - move - split - Type 2
\end{description}

Note, that paths of Type 2 and Type 3 contain at least one split or merge edge, respectively. In the following, we consider only transformation graphs that contain only paths of Type~1--4. 
To identify independent transformation operations, we decompose an optimal transformation graph into its weakly connected components. 
A weakly connected component is a tree if its underlying subgraph is a tree. 

In the following, we give a more substantive view on those weakly connected components which are trees (see Figure \ref{fig:tree_defs}). \\
The first ones are \emph{trees of Type 1}. These trees contain only paths of Type 1, that is, there is only one move edge in the tree connecting one source and sink node (see Figure \ref{fig:tree_defs}a).
A weakly connected component containing only paths of Type 2 has only one source node and at least two paths of Type 2, that is, it has at least two sink nodes. It is a tree since all all nodes have indegree 1. 
We call these trees \emph{trees of Type 2} (see Figure \ref{fig:tree_defs}b). 
\emph{Trees of Type 3} contain only paths of merge or move operations (Type 3). These trees have at least two source nodes whose paths reach the same sink node. All nodes have outdegree 1 (see Figure \ref{fig:tree_defs}c). 
The last weakly connected component which is a tree is a \emph{tree of Type 4}.
These trees include only paths of Type 4. 
They contain only and at least two paths of Type 4, that is, that they have at least two source and two sink nodes (see Figure \ref{fig:tree_defs}d). 
For the following sections, we need a more detailed description of Type-4-Trees.
All source nodes merge and move to some intermediate node $\sigma$.
After $\sigma$, there is a move to $\sigma^*$ with subsequent split and move edges leading to the source nodes. 
All source and sink nodes of this tree are connected by one single path which we call a \emph{bottleneck} with $\sigma$  as the \emph{first bottleneck node} and $\sigma^*$ as the \emph{second bottleneck node}.
All nodes above and including the first bottleneck node have outdegree 1. We call this subgraph the \emph{upper tree} of $\sigma$. 
All nodes below and including the second bottleneck node have indegree 1. We call this subgraph a \emph{lower tree} of $\sigma^*$.  \\

\begin{figure}
    \centering
    \includegraphics[width=0.7\textwidth]{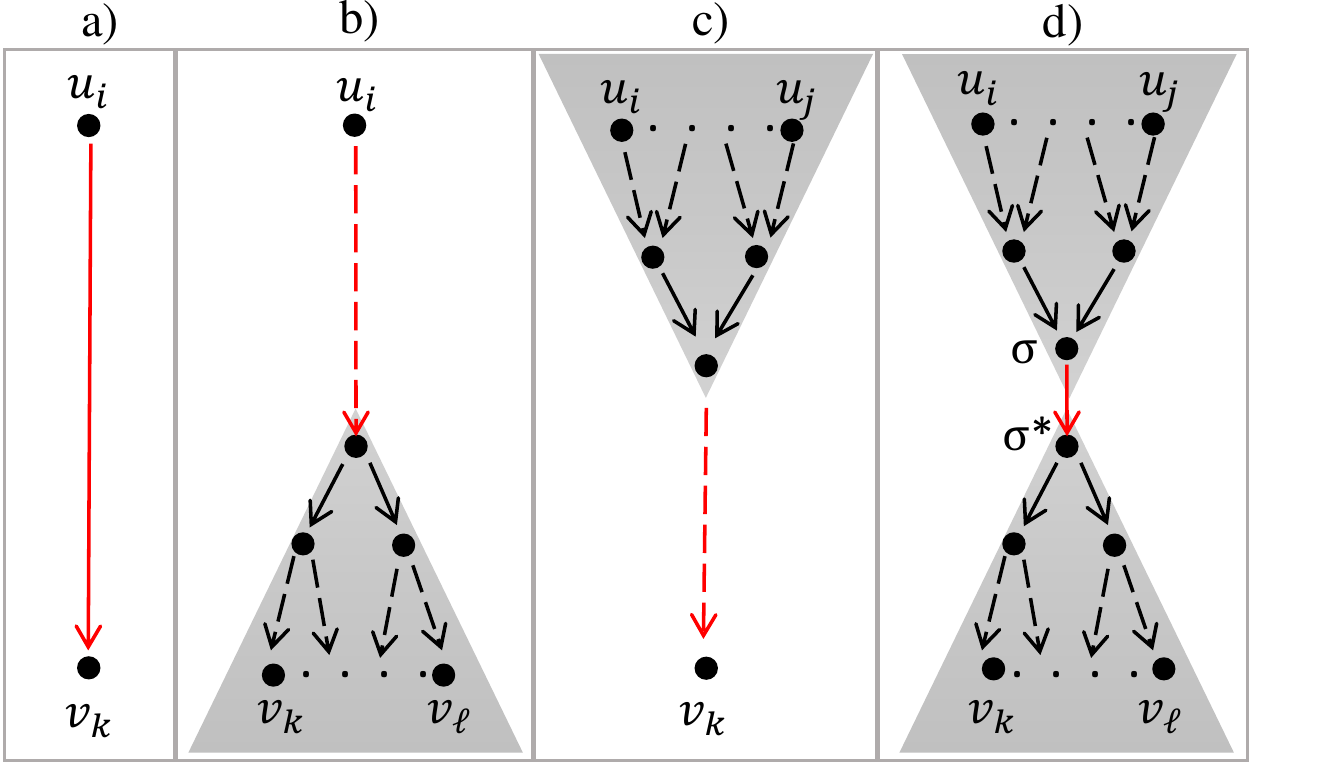}
    \caption{All red edges are move operations. Black arrows are merge or split edges. The dashed lines represent paths from one node to another without specifying how many intermediate node are on them.  a) Tree of Type 1. b) Tree of Type 2. c) Tree of Type 3. d) Tree of Type 4.}
    \label{fig:tree_defs}
\end{figure}

The following lemma states that there always exists an optimal transformation graph, where every weakly connected component is a tree of Type~1--4. 

\begin{lemma}
Let $x$ and $y$ be two time series. Then there exists an optimal transformation graph $G_{\mathbb{S}}(x,y)$ such that its weakly connected components are only trees of Type~1--4. 
\end{lemma}

\begin{proof}
We show that if a weakly connected component of a given optimal transformation graph $G$ is not a tree of Type~1--3, then it has to be a tree of Type~4. 
We consider a path of Type~4. 
In a path of Type~4 there is one part with consecutive merge-move-split edges.  
Let $\sigma$ be the node between this merge and move operation and $\sigma^*$ be the node between this move and split operation.
We add further move and merge operations to the part above $\sigma$. 
We still have outdegree 1 of each node above~$\sigma$. 
The same applies for the part below $\sigma^*$: Adding further move and split operations still leads to indegree 1 of each node below $\sigma^*$. 
Hence, the subgraph of~$G$ consisting of the edge from $\sigma$ to $\sigma^*$ and all move or merge edges above $\sigma$ and all move or split edges below $\sigma^*$ is a tree of Type~4. 
We now show that this tree structure cannot be extended without violating our assumptions of the above Lemmas. 
Let $\alpha$ be a node in the upper tree that is connected to a node~$\alpha'$ that is neither in the upper nor in the lower tree. 
If $\alpha$ is a source node or an intermediate node except of $\sigma$, the first operation is a split, where one split edge is on the path to $\sigma$ and the other is on the path to $\alpha'$. 
We get a contradiction to Lemma \ref{prop:split-move-merge}, because the first path includes consecutive split-move-merge edges. 
If $\alpha=\sigma$, we have again a split at $\alpha$, which is a contradiction to Lemma \ref{prop:merge_split} because we have a consecutive merge-split edge. 
The same argumentation is applied for an extension of the lower tree, because it is the symmetric case of the one we described. 
\end{proof}

It follows that we can decompose an optimal transformation graph $G_{\mathbb{S}}(x,y)$ into a sequence of distinct trees $(\mathcal{T}_1,\ldots,\mathcal{T}_t)$.
Each tree $\mathcal{T}_i$ has a set of sink nodes $N_{\mathcal{T}_i}(x)$ and a set of source nodes $N_{\mathcal{T}_i}(y)$.
All nodes of $N_{\mathcal{T}_i}(x)$  and $N_{\mathcal{T}_i}(y)$ are successors of $N_{\mathcal{T}_{i-1}}(x)$ and $N_{\mathcal{T}_{i-1}}(y)$, respectively.
We call a tree \emph{monotonic} if all paths in the tree are monotonic. 
Further a tree may be specified as \emph{increasing} or \emph{decreasing}. 
Two trees are \emph{equivalent} if they have the same set of source and sink nodes. 
The \emph{cost of a tree $\mathcal{T}$} is the sum of the cost of all edges in the tree.\\

In the following, we denote an optimal transformation graph fulfilling all the above properties as an \emph{optimal transformation forest}.

\section{Properties of the MSM Metric}
\label{chap:properties_MSM_Metric}
As a main result of this section, we prove that for a set of time series $X$ there exists a mean $m$ such that all points of $m$ are points of at least one time series of $X$. 
To this end, we first analyze the structure of trees of optimal transformation forests.
Some of the following results are only proven for trees of Type~4 since these trees include all types of possible paths; as a consequence the proofs for other tree types are simpler versions of the ones for Type~4.

\subsection{Properties of Alignment Trees}
We first regard some properties of so-called \emph{subtrees}, which are substructures of trees of Type 4.  
\subsubsection{Subtrees}
\label{section:subtree}
Let  $G_{\mathbb{S}}(x,y)$ be an optimal transformation forest.
For an intermediate node $\delta$ in $G$, that has two parent nodes connected to it by a merge edge each, let $\mathcal{S}(\delta)$ be the \emph{subtree of} $\delta$ consisting of all source nodes of $G$ that have a path to~$\delta$ and of all nodes and edges on these paths. 
Each subtree has a set of source nodes $N_{\mathcal{S}(\delta)}(x)$.
Let $N_{\mathcal{S}(\delta)}(x) = \{u_i,\ldots, u_j\}$ be the source nodes of~$\mathcal{S}(\delta)$; we call $u_i$ the \emph{start node of} $\mathcal{S}(\delta)$ and $u_j$ the \emph{end node of} $\mathcal{S}(\delta)$.  
A subtree is \emph{increasing} (\emph{decreasing}) if all paths to $\delta$ are increasing (decreasing). 
In the following, we will give some properties of subtrees. If there are two move edges to some nodes $\alpha_2$ and $\beta_2$ that merges to another node $\gamma$ (see Figure \ref{fig:my_subtree_parallel_edges}a), we first observe that these two move edges cannot be both increasing or decreasing.

\begin{lemma}
\label{lemma_parallel_edges}
Let $G_{\mathbb{S}}(x,y)$ be an optimal transformation forest with nodes $\alpha_1, \alpha_2, \beta_1, \beta_2$ and $\gamma$ and move edges between  $\alpha_1$, $\alpha_2$ and $\beta_1$, $\beta_2$. 
If $\alpha_2$ and $\beta_2$ merge to $\gamma$, then the edges between $\alpha_1$, $\alpha_2$ and $\beta_1$, $\beta_2$ cannot be both increasing or decreasing. 
\end{lemma}

\begin{figure}[t]
    \centering
    \includegraphics[width = 0.5\textwidth]{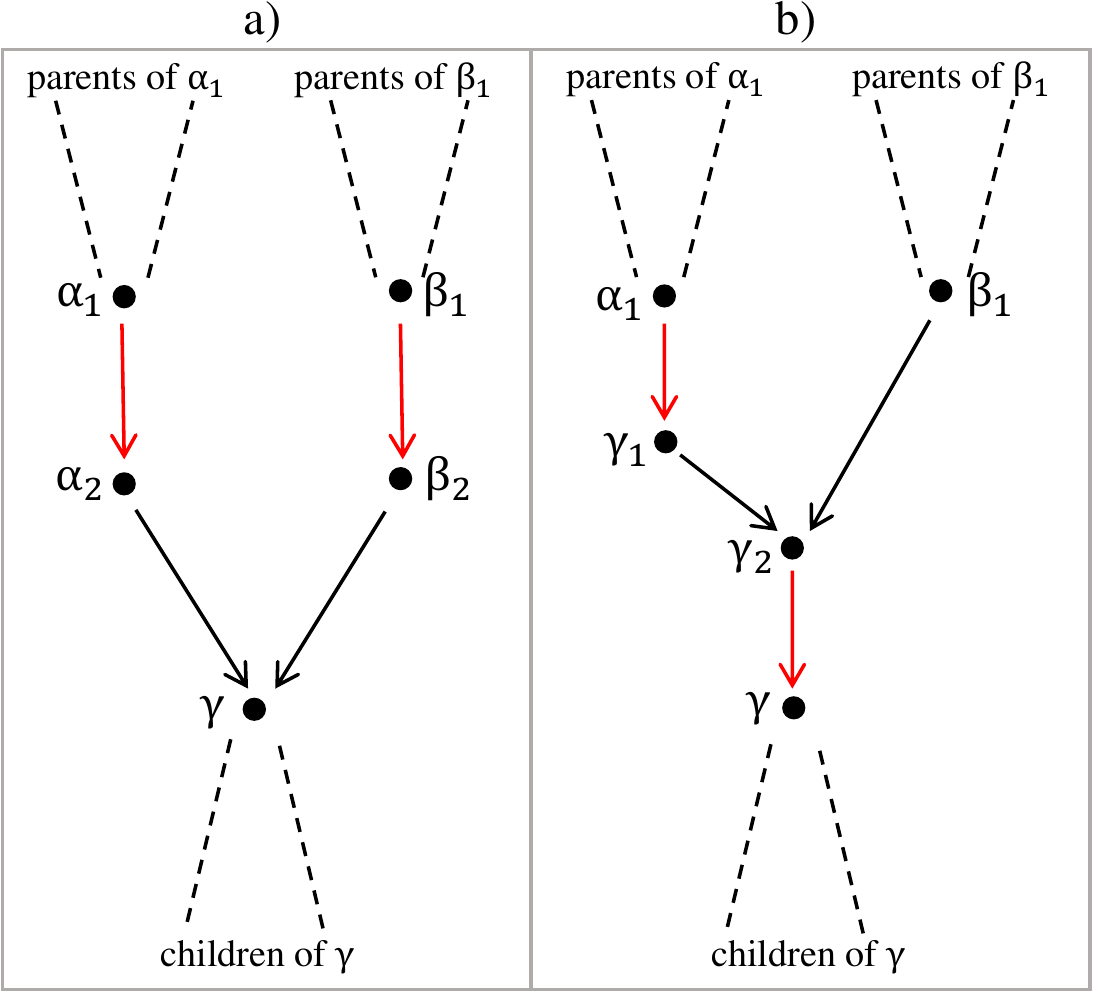}
    \caption{a) Merge  structure of the (intermediate) nodes $\alpha_1, \alpha_2, \beta_1, \beta_2$ and $\gamma$. b) Equivalent tree to the tree in a).}
    \label{fig:my_subtree_parallel_edges}
\end{figure}

\begin{proof}
Without loss of generality, we prove this assumption for increasing paths. Assume towards a contradiction two inc-edges between  $\alpha_1$ and $\alpha_2$ and between $\beta_1$ and $\beta_2$ with $\val(\alpha_1) < \val(\beta_1)$ (see Figure \ref{fig:my_subtree_parallel_edges}a). 
Since $\val(\alpha_2) = \val(\beta_2) = \val(\gamma)$ the cost for these move operations are $ 2\val(\gamma) - \val(\beta_1) - \val(\alpha_1)$.
We now consider a modified merge structure with an additional intermediate node~$\gamma_1$ with $\val(\gamma_1)~=~\val(\beta_1)$ (see Figure \ref{fig:my_subtree_parallel_edges}b). 
We now have an inc-edge from $\alpha_1$ to~$\gamma_1$, which merges with the new node~$\beta_1$ to a new node~$\gamma_2$. 
At $\gamma_2$, there is an inc-edge to $\gamma$. 
The modified transformation forest is equivalent to the old one, since the parent and children nodes of the regarded nodes stay the same (see Figure \ref{fig:my_subtree_parallel_edges}). 
The cost for the modified move operations are $\val(\gamma) - \val(\beta_1) + \val(\beta_1) - \val(\alpha_1) < 2\val(\gamma) - \val(\beta_1) - \val(\alpha_1)$ since $\val(\beta_1) < \val(\gamma)$. 
This is a contradiction to $G$ being optimal. 
\end{proof}

We now make two observations about the value of the node~$\delta$ in a subtree $\mathcal{S}(\delta)$. The first lemma states, that the value of $\delta$ is equal to one value of the source nodes of $\mathcal{S}(\delta)$. Recall that $u_1,\ldots,u_m$ are the source nodes of $G_{\mathbb{S}}(x,y)$ with values $x_1,\ldots, x_m$. 

\begin{lemma}
\label{lemma_max}
Let $\mathcal{S}(\delta)$  be an  increasing (decreasing)  subtree of $\delta$ in an optimal transformation forest $G_{\mathbb{S}}(x,y)$  with $N_{\mathcal{S}(\delta)}(x)~=~\{u_i, \ldots, u_j\}$. 
Then, $\val(\delta) = \max(x_i, \ldots, x_j)$ for increasing subtrees and  $\val(\delta) = \min(x_i, \ldots, x_j)$ for decreasing subtrees. 
\end{lemma}

\begin{figure}[t]
    \centering
    \includegraphics[width=0.45\textwidth]{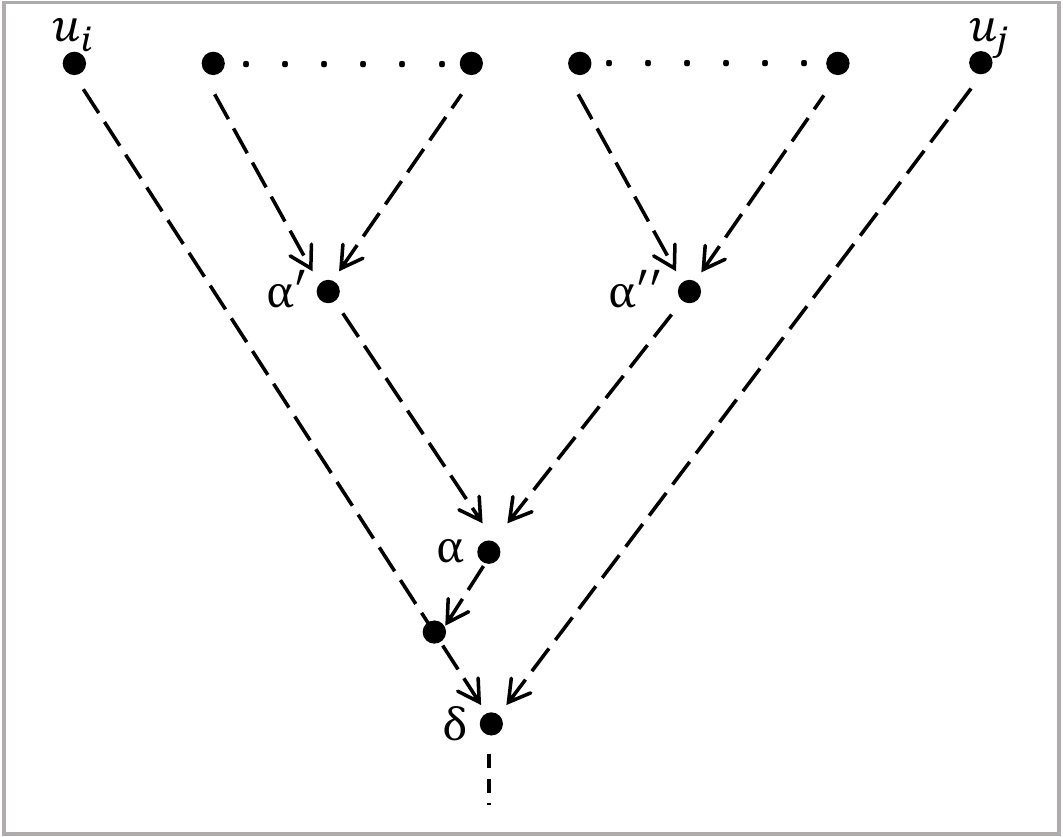}
    \caption{Structure of the subtree of $\delta$ explained in the proof of Lemma \ref{lemma_max}. Note, that this is only a schematic representation and that there may be further intermediate nodes which are not marked.}
    \label{fig:delta_max}
\end{figure}

\begin{proof}[Proof of Lemma \ref{lemma_max}]
Without loss of generality, we prove this assumption for increasing subtrees. Figure \ref{fig:delta_max} depicts this subgraph with all mentioned intermediate nodes. Assume towards a contradiction, that  $\val(\delta) \neq \max(x_i, \ldots, x_j)$. Therefore, there exists an $u_{\ell} \in \{u_i, \ldots, u_j\}$ such that $\val(\delta)\neq x_\ell$.
For $\val(\delta)< x_\ell$, it follows that we have a decreasing edge between $u_{\ell}$ and $\delta$, which is a contradiction.
For $\val(\delta)> x_\ell$, let $\alpha$ be the first intermediate node below $\{u_i, \ldots, u_j\}$ such that $\val(\alpha) \neq \max_{u\in \mathcal{N}_{\mathcal{S}(\alpha)}(x)} (u)$, where $\mathcal{N}_{\mathcal{S}(\alpha)}(x) \subseteq \mathcal{N}_{\mathcal{S}(\delta)}(x)$ are the source nodes of the subtree $\mathcal{S}(\alpha)$ of $\alpha$.
There exist two intermediate nodes  $\alpha'$ and $\alpha''$ such that for the source nodes of their subtrees $\mathcal{S}(\alpha')$ and  $\mathcal{S}(\alpha'')$, respectively, it holds that $\mathcal{N}_{\mathcal{S}(\alpha')}(x) \cup \mathcal{N}_{\mathcal{S}(\alpha'')}(x) = \mathcal{N}_{\mathcal{S}(\alpha)}(x)$. 
It follows that there exists a path from $\alpha'$ to $\alpha$ and from $\alpha''$ to~$\alpha$. Since $\alpha$ is the first intermediate node below $\mathcal{N}_{\mathcal{S}(\delta)}(x)$ such that $\val(\alpha)\neq \max_{u\in \mathcal{N}_{\mathcal{S}(\delta)}(x)} (u)$, it holds that $\val(\alpha') = \max_{u\in \mathcal{N}_{\mathcal{S}(\alpha')}(x)} (u)$ and $\val(\alpha'') = \max_{u\in \mathcal{N}_{\mathcal{S}(\alpha'')}(x)} (u)$. 
Consequently, there is an inc-edge on the path from~$\alpha'$ to $\alpha$ and on the path from $\alpha''$ to $\alpha$.
Applying Lemma \ref{lemma_parallel_edges}, this is a contradiction to $G$ being optimal.
\end{proof}

In the next lemma, we specify the value of $\delta$ in a subtree $\mathcal{S}(\delta)$, stating that it is always equal to the value of a specific source node of $\mathcal{N}_{\mathcal{S}(\delta)}(x)$. 

\begin{lemma}
\label{lemma:max_specification}
Let the nodes $\delta_1$ and $\delta_2$ merge to a node $\delta$ in an optimal transformation forest $G$. Let $\mathcal{S}(\delta)$ be the subtree of $\delta$, $\mathcal{S}(\delta_1)$ be the subtree of $\delta_1$ with the end node $u_{i-1}$, and $\mathcal{S}(\delta_2)$ be the subtree of $\delta_2$ with the start node $u_i$.
If the subtree $\mathcal{S}(\delta)$ is increasing (decreasing) and $x_{i-1}>x_i$, then $\val(\delta)=x_{i-1}$ $(\val(\delta)=x_i)$. If $x_{i-1}<x_i$, then $\val(\delta)=x_i$ ($\val(\delta)=x_{i-1}$ for decreasing subtrees).
\end{lemma}

\begin{figure}[h]
    \centering
    \includegraphics[width=0.35\textwidth]{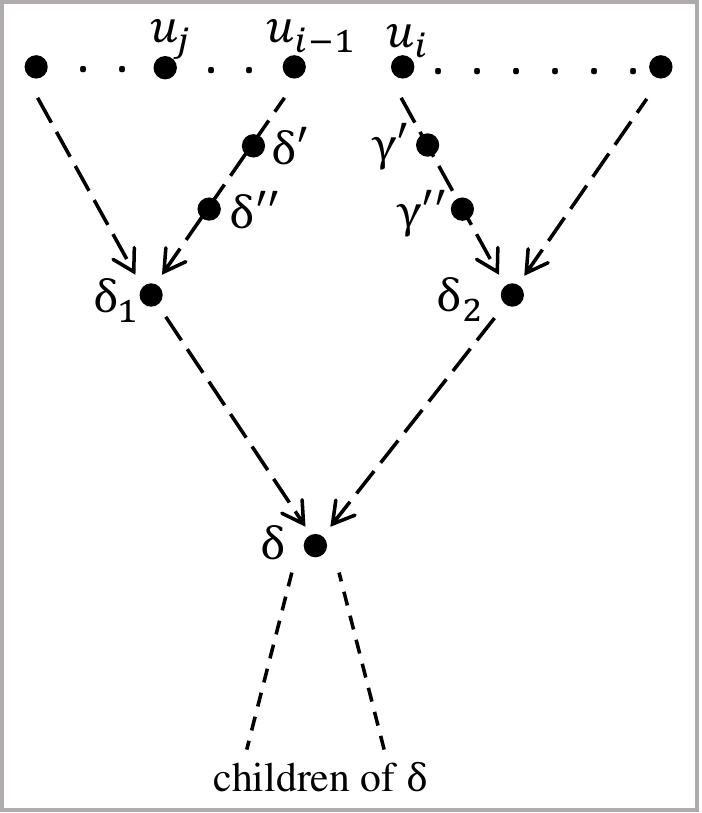}
    \caption{Structure of the subtree of $\delta$ described in the proof of Lemma~\ref{lemma:max_specification}. Note that this is only a schematic representation and that there may be further intermediate nodes which are not marked.}
    \label{fig:max_specification}
\end{figure}

\begin{proof}
Without loss of generality, $\mathcal{S}(\delta)$ is increasing.
Figure \ref{fig:max_specification} depicts this subgraph with all mentioned intermediate points. 
We will only show the case that $x_{i-1}~>~x_i$ since the other case is analogous. 
By Lemma \ref{lemma_max}, it holds that $\val(\delta)~=~\max(\val(\delta_1), \val(\delta_2))$. 
We first show that  $\val(\delta)~=~\val(\delta_1)$. 
Assume towards a contradiction that $\val(\delta_1) < \val(\delta_2) = \val(\delta)$. 
Since $x_{i-1} > x_i$, there exists intermediate nodes $\gamma'$ and $\gamma''$ on the path between $u_i$ and $\delta_2$ such that $x_{i-1} \geq \val(\gamma')$ and $x_{i-1} < \val(\gamma'')$.
Let $\mathcal{S'}(\delta)$ be the modified subtree of $\mathcal{S}(\delta)$. 
The only difference between $\mathcal{S}$ and $\mathcal{S}'$ is that $\gamma'$ merges to some intermediate node on the path between $u_{i-1}$ and $\delta_1$.
The cost of $\mathcal{S}'(\delta)$ is $\cost(\mathcal{S}'(\delta))~=~\cost(\mathcal{S}(\delta)) - \val(\gamma'') - \val(\gamma') < \cost(\mathcal{S}(\delta))$.
This is a contradiction to $G$ being optimal. 
Applying Lemma~\ref{lemma_parallel_edges} and Lemma~\ref{lemma_max}, we get $\val(\delta) = \val(\delta_1)$.

In a second step, we prove that $\val(\delta_1) = x_{i-1} = \max_{u\in\mathcal{N}_{\mathcal{S}(\delta_1)}(x)}(\val(u))$. 
Assume towards a contradiction, that there exists a~$u_j \in N_{\mathcal{S}((\delta_1)}(x)\setminus u_{i-1}$ such that $\val(\delta_1)~=~x_j > x_{i-1} > x_i$.
Then, it holds that there exist two intermediate nodes $\delta'$ and $\delta''$ on the path between $u_{i-1}$ and $\delta_1$ such that $\val(\delta')<\val(\delta_1)$ and $\val(\delta'')~=~\val(\delta_1)$. 
We consider the modified subtree $\mathcal{S}''(\delta)$, which is almost equal to  $\mathcal{S}(\delta)$, the only difference being that $\delta_2$ merges to some intermediate node on the path between $\delta'$ and $\delta''$.
The cost of $\mathcal{S}''(\delta)$ are $\cost(\mathcal{S}''(\delta)) = \cost(\mathcal{S}(\delta)) - \val(\delta_1) - \val(\delta_2) < \cost(\mathcal{S}(\delta))$. 
This is a contradiction to $G$ being optimal. 
\end{proof}

We will now apply the above properties to the bottleneck nodes in a tree of Type 4 stating that the first and second bottleneck nodes always have values of the input time series $x$ and $y$, respectively. Recall, that the first bottleneck node $\sigma$ is the intermediate node where all source nodes in the tree of Type~4 merge to, followed by a move edge to the second bottleneck node $\sigma^*$. 

\begin{corollary}
\label{lemma:sigma}
Let $G_{\mathbb{S}}(x,y)$ be an optimal transformation forest. 
In a tree $\mathcal{T}$ of Type~4 $\val(\sigma)\in V_{\mathcal{T}}(x)$ and $\val(\sigma^*)\in V_{\mathcal{T}}(y)$.
\end{corollary}

\begin{proof}
To prove that $\sigma\in V_{\mathcal{T}}(x)$ we apply Lemma \ref{lemma_max} since the upper part of the tree $\mathcal{T}$ is the subtree of $\sigma$. 
By symmetry reasons, it follows that $\sigma^* \in  V_{\mathcal{T}}(y)$.
\end{proof}

\subsection{The Effect of Perturbing Single Values}
We aim to show that there exists a mean of a set of time series that only consists of points of the input set. 
To this end, we first make observations on the effect of shifting points of a time series that are not from $V(X)$.
To this end, we first analyze for two time series $x$ and $y$, how the distance between $x$~and~$y$ may be affected by shifting one point of $x$ by $\varepsilon\in \mathds{R}$.
We let $x_{\varepsilon, i}$ denote the new time series that is equal to $x$ except at the position $i$, where it has the new point $x_i +\varepsilon$. 
The change of the node $u_i$ in the transformation forest is denoted by $u_i^\varepsilon$. 
In the following we say that if the distance between $x_{\varepsilon, i}$ and $y$ is shorter than between $x$ and $y$, the replacement of $x$ by $x_{\varepsilon, i}$ is \emph{beneficial}. If it leads to a longer distance, it is \emph{detrimental}, and if the distance does not change it is \emph{neutral}. 
Assume that $x_i\notin V(y)$, the next lemma states that if the replacement of $x$ by $x_{\varepsilon, i}$ is not neutral, it is beneficial for either $\varepsilon$ or $-\varepsilon$.

\begin{lemma}
\label{lemma:deviation_epsilon}
 Let $x$ and $y$ be two time series with distance $d(x,y)$. If $x_{i-1} \neq x_i \neq x_{i+1}$ and $x_i \notin V(y)$, there either exists an $\varepsilon'>0$ such that for all $\varepsilon \in [0,\varepsilon']$ one of the following equations holds:
\begin{description}
    \item[(1)] $ d(x_{\varepsilon, i}, y) + \varepsilon   = d(x,y)  = d(x_{-\varepsilon, i},y) -\varepsilon $ (beneficial increase),
    \item[(2)] $ d(x_{-\varepsilon, i}, y) + \varepsilon   = d(x,y)  = d(x_{\varepsilon, i},y) -\varepsilon $ (beneficial decrease),
\end{description}
or there exist $\varepsilon_I,\varepsilon_D>0$ such that  
\begin{description}
    \item[(3.1)]  $d(x,y) = d(x_{\varepsilon, i},y)$ for all $\varepsilon \in [0,\varepsilon_I]$ (neutral increase), and
    \item[(3.2)]  $d(x,y) = d(x_{-\varepsilon, i},y)$ for all $\varepsilon \in [0,\varepsilon_D]$ (neutral decrease).
\end{description}
Moreover, for beneficial increases $x_i+\varepsilon' \in (V(y) \cup \{x_{i-1},x_{i+1}\})$,
for beneficial decreases $x_i-\varepsilon' \in (V(y) \cup \{x_{i-1},x_{i+1}\})$, and for neutral increases and decreases  $x_i+\varepsilon_I , x_i-\varepsilon_D \in (V(y) \cup \{x_{i-1},x_{i+1}\})$.
\end{lemma}

\begin{proof}
We show the lemma for trees of Type~4. All other cases are simpler versions of this proof.
Let $\mathcal{T}$ be a tree of Type~4 in $G_{\mathbb{S}}(x,y)$. 
By Lemma \ref{lemma:monotonicity}, the tree $\mathcal{T}$ is monotonic. 
We assume, without loss of generality, that all monotonic paths in $\mathcal{T}$ are increasing. 
We distinguish whether $u_i$ has only predecessors or only successors (\textit{Case 1}) or both (\textit{Case~2}) in $\mathcal{T}$. We denote the predecessors of $u_i$ as $\mathcal{P}$ and the successors of $u_i$ as $\mathcal{F}$.

\indent \textit{Case 1}: $u_i$ has only predecessors or successors in $\mathcal{T}$. 
We prove the case that  $u_i$ has only predecessors, the other case is analogous. 
We first describe the possible structures of the upper tree in $\mathcal{T}$ for this case, which are depicted in Figure \ref{fig:epsilon_var}. 
There is a potential move at $u_i$ to a node~$\gamma$. 
The node $\gamma$ merges to $\delta$ with a node $\alpha^*$, which is the node resulting from a move at $\alpha$. 
The nodes $\{u_{i-\ell}, \ldots, u_{i-1}\} \subseteq \mathcal{P}$ are the source nodes of the subtree of $\alpha$.
Below $\delta$ there may be further subsequent merge and move operations to the first bottleneck node $\sigma$. 
Since $x_{i-1} \neq x_i$ there has to be an inc-edge either between $u_i$ and $\gamma$, if  $x_{i-1} > x_i$, or between $u_{i-1}$ and $\alpha^*$, if $x_{i-1} < x_i$ because in the first case $\val(\delta) = x_{i-1}$ and in the second case $\val(\delta) = x_i$ (see Lemma \ref{lemma:max_specification}). 

\indent\textit{Case 1.1}: $x_{i-1} > x_{i}$.
There is an inc-edge between $u_i$ and $\gamma$ (see Figure \ref{fig:epsilon_var}a). 
The replacement of $x$ by $x_{\varepsilon, i}$ is a beneficial increase for all $\varepsilon \in [0,\varepsilon']$ with $\varepsilon' = \val(\gamma)-x_i$ because the node $u_i^\varepsilon$  approaches the node $\gamma$ and the cost of the adapted move decrease by $\varepsilon$.   
Thus, we get the left side of Equation (1), $ d(x_{\varepsilon, i}, y) + \varepsilon = d(x,y)$. 
Since the subtree of $\delta$ is increasing and $x_{i-1}>x_i$, it holds by Lemma \ref{lemma:max_specification} that $\val(\delta) = x_{i-1}$. 
We get that $x_{i-1} = \val(\gamma) = x_i + \varepsilon'$. 
For the right side of Equation (1), the argumentation is similar: After replacing  $x$ by $x_{-\varepsilon, i}$ for $\varepsilon\leq \varepsilon'$, the cost for the move between $x_i - \varepsilon$ and $\gamma$ are $\val(\gamma) - x_i + \varepsilon$.
Therefore, they increase by $\varepsilon$.

\indent\textit{Case 1.2}: $x_{i-1} < x_{i}$.
There is an inc-edge between $u_{i-1}$ and $\alpha^*$ (see Figure \ref{fig:epsilon_var}b). 
We modify the structure of $\mathcal{T}$ for the replacement of $x$ by $x_{\varepsilon, i}$ for $\varepsilon \in [0,\varepsilon_I]$, $\varepsilon_I>0$.
Let $\mathcal{T}'$ be the modified tree with a new node $u_i^\varepsilon$ instead of $u_i$. 
In $\mathcal{T}'$, the nodes $\alpha^*$ and $\delta$ does not exist but $\mathcal{T}'$ contains a new node $\delta'$ such that $\val(\delta') \in [\val(\alpha^*),\val(\sigma^*)]$. The node~$u_i^\varepsilon$ merges to $\delta'$.
The rest of the tree stays unchanged. 
For all $\varepsilon \in [0,\varepsilon_I]$ with $\varepsilon_I = \val(\sigma^*)-x_i$ the cost of $\mathcal{T}'$ is equal to the cost of $\mathcal{T}$ because we only shifted a merge operation to another position in the tree (see Figure \ref{fig:epsilon_var}c). 
This is a neutral increase for all $\varepsilon \in [0,\varepsilon_I]$.
It holds that $x_i + \varepsilon_I = \sigma^* \in V(y)$ (see Corollary~\ref{lemma:sigma}). 
Let $\mathcal{T}''$ be another modified tree of $\mathcal{T}$ with a new node $u_i^{-\varepsilon}$ instead of $u_i$ for $\varepsilon \in [0,\varepsilon_D]$, $\varepsilon_D>0$.
The tree $\mathcal{T}''$ does not contain the node $\delta$ but contains a new node $\delta''$ such that $\val(\delta'') \in [x_{i-1}, \val(\alpha^*)]$ and $u_i^{-\varepsilon}$ merges to $\delta''$ (see Figure \ref{fig:epsilon_var}d). 
For all $\varepsilon \in [0,\varepsilon_D]$ with $\varepsilon_D = \val(\alpha^*)-x_{i-1} $ we get equal cost of $\mathcal{T}$ and $\mathcal{T}''$ since we only shifted a merge operation.
From Lemma \ref{lemma:max_specification} we get that $\val(\alpha^*) = \val(\delta) = x_i$ and hence $x_i- \varepsilon_D= x_{i-1}$.

\begin{figure}[h]
    \centering
    \includegraphics[width = 0.8\textwidth]{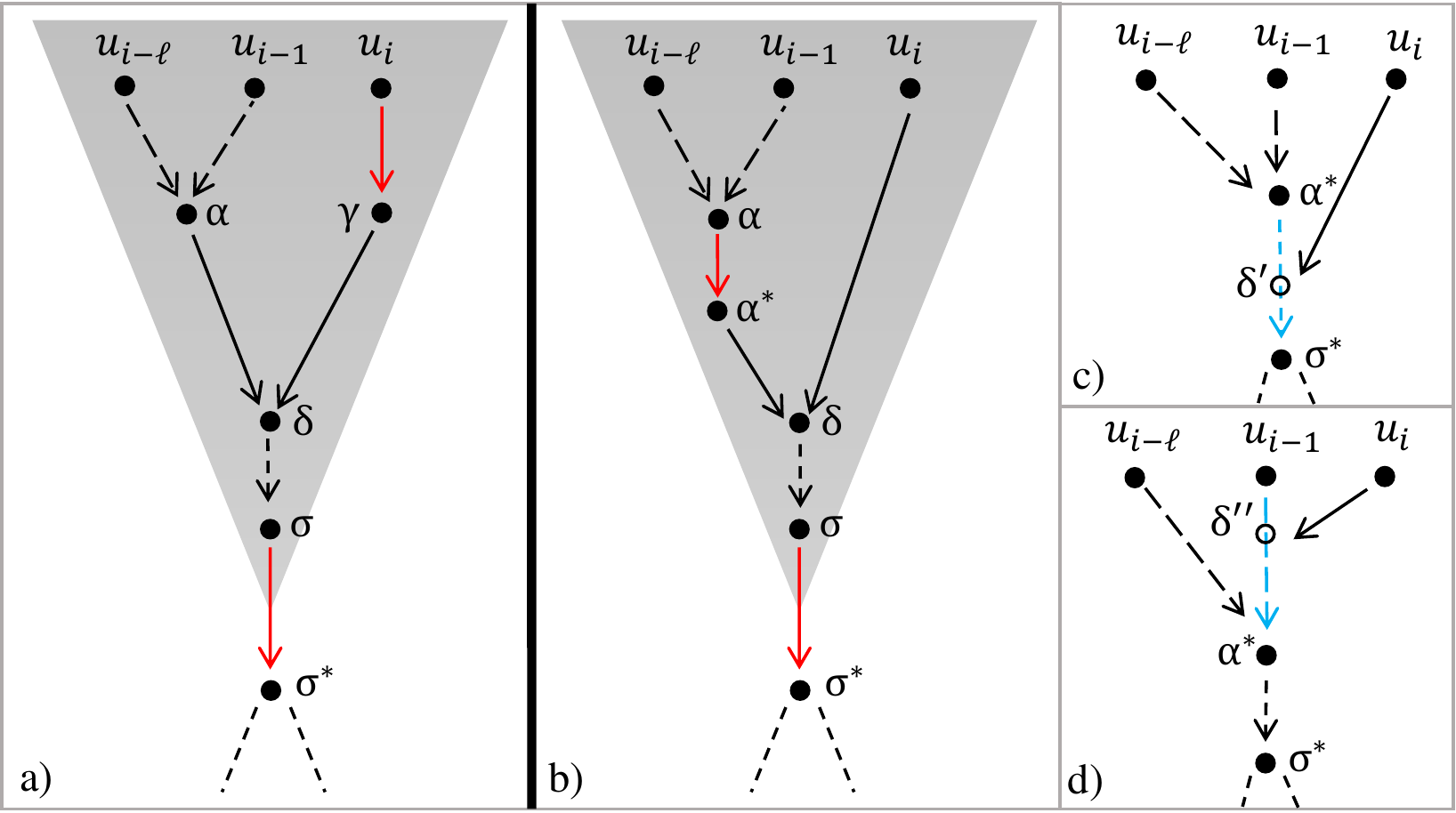}
    \caption{Schematic representation of the trees discussed for Case 1 in the proof of Lemma \ref{lemma:deviation_epsilon}. The node $u_i$ has only predecessors. The dashed red edges show potential move operations. a) Case 1.1:  $x_{i-1} > x_{i}$. b) Case 2.1: $x_{i-1} < x_{i}$. c) Proof mechanism introducing the modified tree $\mathcal{T}'$, where the path on which the new node $\delta'$ can be shifted on is marked in blue. d) Proof mechanism introducing the modified tree $\mathcal{T}''$ following the same mechanism as in c).}
    \label{fig:epsilon_var}
\end{figure}

\indent\textit{Case 2}: 
$u_i $ has predecessors $\mathcal{P}$ and successors $\mathcal{F}$. Again, we first describe the upper Type-4-Tree $\mathcal{T}$ (see Figure \ref{fig:epsilon_var_II}).  
Let $\{u_{i-\ell}, \ldots, u_{i-1}\}\subseteq \mathcal{P}$ be the source nodes of the subtree of $\alpha$. 
At $\alpha$ there is a potential move to $\alpha^*$. 
Let $\{u_{i+1}, \ldots, u_{i+r}\}\subseteq \mathcal{F}$ be the source nodes of the subtree of  $\zeta$. At $\zeta$ there is a potential move to $\zeta^*$ 
After a potential move from $u_i$ to $\gamma$ there is a merge with $\alpha^*$, which is afterwards merged with $\zeta^*$ to an intermediate node $\delta$. 
Without loss of generality, we assume this order of merge to $\delta$. What follows are potential move and merge operations until all nodes in $N_{\mathcal{T}}(x)$ merge to the first bottleneck node $\sigma$.
Since $\mathcal{T}$ is increasing, the subtree of $\delta$ is increasing. 
We further analyze the relation between $x_i$ to its direct predecessor $x_{i-1}$ and its direct successor $x_{i+1}$.

\indent\textit{Case 2.1}: $x_{i-1} < x_i < x_{i+1}$. 
By Lemma \ref{lemma:max_specification}, it follows that $\val(\delta) =x_{i+1}$.
Furthermore, there is no inc-edge between $u_i$ and $\gamma$ because $u_i$ merges with $\alpha^*$ to $\beta$ with a subsequent inc-edge to $\beta^*$ (see Figure \ref{fig:epsilon_var_II}a).
We modify the tree structure of $\mathcal{T}$ for the replacement of $x$ by $x_{\varepsilon, i}$.
Let $\mathcal{T}'$ be the modified tree of $\mathcal{T}$, where we have the new node $u_i^\varepsilon$ instead of $u_i$ for $\varepsilon \in [0,\varepsilon_I], \varepsilon_I>0$. 
In $\mathcal{T}'$ the node $\beta$ does not exist anymore but $\mathcal{T}'$ includes a new node $\beta'$, such that $\val(\beta') \in [\val(\alpha^*),\val(\beta^*)]$, where $u_i^\varepsilon$ merges with $\alpha^*$. 
The rest of the tree stays unchanged. 
For $\varepsilon \in [0,\varepsilon_I]$ with  $\varepsilon_I = \val(\beta^*)-x_i$ the cost of $\mathcal{T}'$ is equal to the cost of $\mathcal{T}$ because we only shifted a merge operation to another position in the tree. 
This is a neutral increase for all $\varepsilon \in [0,\varepsilon_I]$.
It holds that $x_i +\varepsilon_I = \val(\beta^*) = x_{i+1}$.
Let further be $\mathcal{T}''$ another modified tree of $\mathcal{T}$. 
The tree $\mathcal{T}''$ does not contain the node $\beta$, instead it contains a new node $\beta''$, such that $\val(\beta'') \in [u_{i-1},\val(\alpha^*)]$, where $u_i^{-\varepsilon}$ merges to. 
Again, we only shifted a merge operation, that leads to equal cost of $\mathcal{T}$ and $\mathcal{T}''$ for all $\varepsilon \in [0,\varepsilon_D]$ with $\varepsilon_D = \val(\alpha^*)-x_{i-1} $. 
We have $\val(\alpha^*) = x_i$ and hence $x_i - \varepsilon_D = x_{i-1}$.

\indent\textit{Case 2.2}: $x_{i-1} > x_i > x_{i+1}$. This case is analogous to Case 2.1.

\indent\textit{Case 2.3}: $x_{i-1} < x_i > x_{i+1}$. 
We further assume, without loss of generality, that $x_{i-1} < x_{i+1}$.
By Lemma \ref{lemma:max_specification} it holds that $x_{i} = \val(\delta)$.
We have inc-edges between $u_{i-1}$ and $\alpha^*$ and between $u_{i+1}$ and $\zeta^*$ (see Figure \ref{fig:epsilon_var_II}b). 
The replacement of $x$ by $x_{-\varepsilon, i}$ for an $\varepsilon \in [0, \varepsilon_D]$ is a beneficial decrease because the merge points $\beta$ and $\delta$ are shifted by $-\varepsilon$: The move cost are $\val(\alpha^*) - \varepsilon - x_{i-1}$ and $\val(\zeta^*) - \varepsilon - x_{i+1}$ for the two move operations, that is a decrease of $2\varepsilon$.
The new merge node of $\beta^*$ and $\zeta^*$ is denoted by $\delta'$.
For the new path between $\delta'$ and $\sigma^*$ we have cost of $\vert\sigma^* - \delta' + \varepsilon\vert$, that is an increase of cost by $\varepsilon$. 
We get the left side of Equation (2), that is, $ d(x_{-\varepsilon, i}, y) + \varepsilon   = d(x,y)$ for all $\varepsilon \in [0,\varepsilon']$ with $\varepsilon_D = x_i-x_{i+1}$. 
It holds that  $x_i - \varepsilon_D = x_{i+1}$.
The argumentation of the detrimental replacement of  $x$ by $x_{\varepsilon, i}$ is analogous to Case 1.1.

\indent\textit{Case 2.4}: $x_{i-1} > x_i < x_{i+1}$. 
There is an inc-edge between $u_i$ and $\gamma$ (see Figure \ref{fig:epsilon_var_II}c). 
Again, we further assume, without loss of generality, that $x_{i-1} < x_{i+1}$.
Following the same argumentation as in Case 1.1, the replacement of $x$ by $x_{\varepsilon, i}$ is a beneficial increase for all $\varepsilon \in [0,\varepsilon_I]$ with $\varepsilon_I = \val(\gamma)-x_i$.
By Lemma \ref{lemma:max_specification}, it holds that $\val(\beta) = x_{i-1}$ and $\val(\delta) =  x_{i+1}$. 
Note, that there are no increasing paths between $u_{i-1}$ and $\beta^*$ and $u_{i+1}$ and $\zeta^*$ because otherwise there is no move between $u_i$ and $\gamma$ (see Lemma~\ref{lemma_parallel_edges}). 
It holds that $ \val(\gamma) = x_{i-1} =  x_i+\varepsilon'$.
The detrimental replacement of  $x$ by $x_{-\varepsilon, i}$ is analogous to Case~1.1. 
\end{proof}

\begin{figure}[h]
    \centering
    \includegraphics[width = 0.9\textwidth]{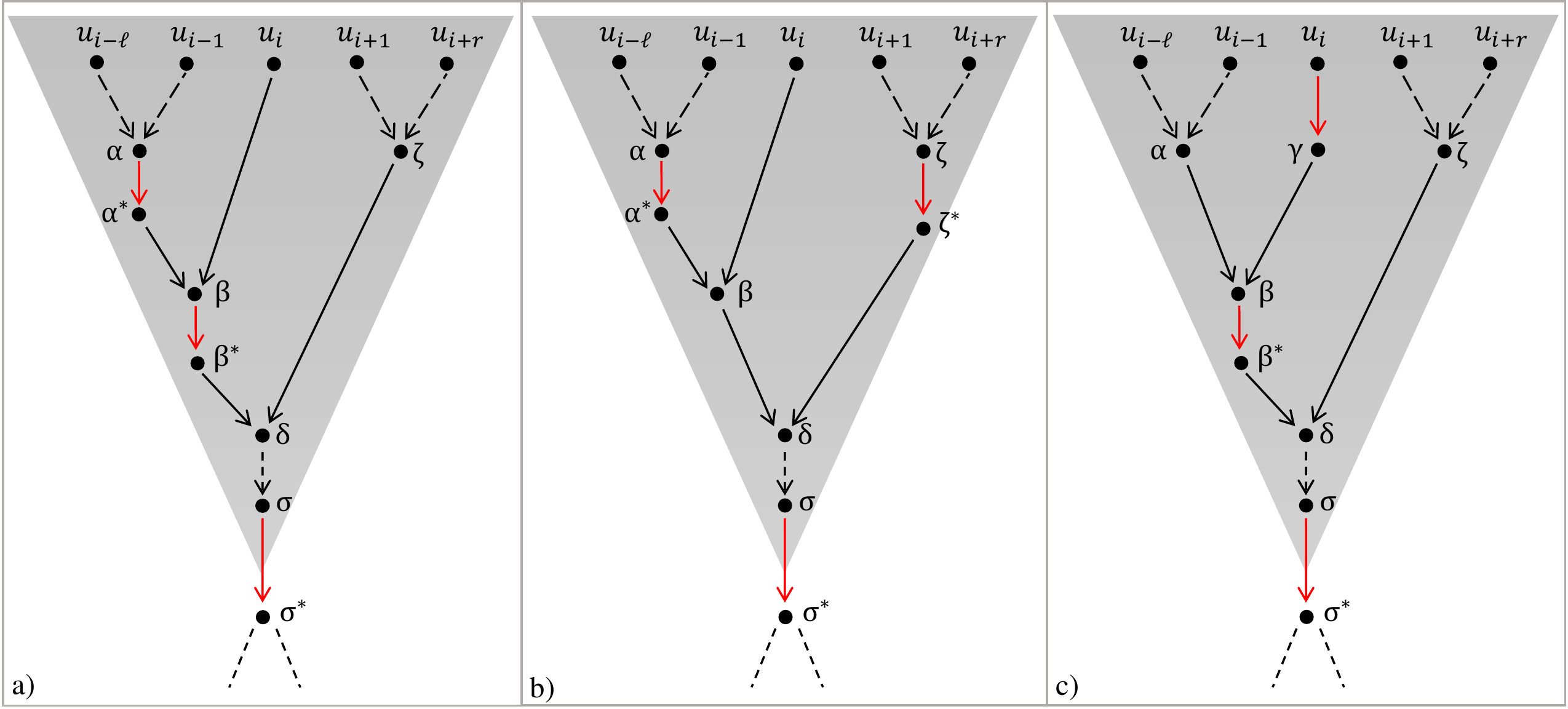}
    \caption{Schematic representation of the trees discussed for Case 2 in the proof of Lemma \ref{lemma:deviation_epsilon}. The node $u_i$ has predecessors and successors. a) Case 2.1: $x_{i-1} < x_i < x_{i+1}$. b) Case 2.3: $x_{i-1} < x_i > x_{i+1}$. c) Case 2.4:  $x_{i-1} > x_i < x_{i+1}$.}
    \label{fig:epsilon_var_II}
\end{figure}

Let us briefly discuss the trees of Type~1--3. If the tree $\mathcal{T}$ is of Type 3, then~$\val(\sigma^*)$ is already in $\mathcal{N}_{\mathcal{T}}(y)$. The same proof as for trees of Type 4 can be applied. 
Since the symmetries properties hold for the MSM metric, the Lemma holds for trees of Type 2 as well. 
For a tree containing only a move edge, the argumentation is the same as in Case 1.1.

In the following, we regard a \emph{block} $\mathcal{B}$ of adjacent source nodes $N_{\mathcal{B}}(x)=\{u_i, \ldots, u_\ell\}$ representing points of equal value of a time series $x$. 
A block is a maximal contiguous sequence of nodes with the same value. 
Our aim is to show a generalization of Lemma \ref{lemma:deviation_epsilon} shifting all points of a block $\mathcal{B}$ by some $\varepsilon\in\mathds{R}$.  
We show that shifting a block is either beneficial for one direction or neutral for both directions.
Let $x_{\varepsilon, i,\ell}$, $i<\ell$, denote the time series that is equal to $x$ except at the positions $i,\ldots,\ell$, where the points $x_i$ of $x$ are replaced by  $x_i + \varepsilon$.
The definitions of beneficial, detrimental or neutral replacements of $x$ by $x_{\varepsilon, i,\ell}$  are analogous to the previous one. 
A block may be contained in several trees, hence shifting a block affects the cost of all these trees. To count the number of trees with beneficial or detrimental replacement, we introduce two further parameters $\rho_I,\rho_D\in\mathds{N}$.
\begin{lemma}

\label{lemma_deviation_epsilon_block}
 Let $x=(x_1,\ldots , x_m)$ and $y=(y_1,\ldots, y_n)$ be two time series with a distance $d(x,y)$. 
 If we consider a block $\mathcal{B}$ of similar points $N_{\mathcal{B}}(x) = \{u_i, \ldots, u_\ell\}$ with  $ x_i \notin V(y)$, then there either exists an $\varepsilon'>0$ and $\rho_I ,\rho_D \in \mathds{N}$ such that for all $\varepsilon \in [0,\varepsilon']$ one of the following equations holds:
 
\begin{description}
    \item[(1)] $ d(x_{\varepsilon, i,\ell}, y) + \rho_I \cdot \varepsilon = d(x,y)  = d(x_{-\varepsilon, i,\ell},y) - \rho_D \cdot\varepsilon $ (b. increase),
    \item[(2)] $ d(x_{-\varepsilon, i,\ell}, y) + \rho_D \cdot\varepsilon  = d(x,y)  = d(x_{\varepsilon, i,\ell},y) - \rho_I \cdot\varepsilon $ (b. decrease),
\end{description}
or there exist $\varepsilon_I,\varepsilon_D>0$ such that
\begin{description}
    \item[(3.1)]  $d(x,y) = d(x_{\varepsilon, i},y)$ for all $\varepsilon \in [0,\varepsilon_I]$ (neutral increase), and
    \item[(3.2)]  $d(x,y) = d(x_{-\varepsilon, i},y)$ for all $\varepsilon \in [0,\varepsilon_D]$ (neutral decrease).
\end{description}
Moreover, for beneficial increases $x_i+\varepsilon' \in (V(y) \cup \{x_{i-1},x_{i+1}\})$,
for beneficial decreases $x_i-\varepsilon' \in (V(y) \cup \{x_{i-1},x_{i+1}\})$, and for neutral increases  and decreases~$x_i+\varepsilon_I, x_i-\varepsilon_D \in (V(y) \cup \{x_{i-1},x_{i+1}\})$.
\end{lemma}

\begin{figure}[h]
    \centering
    \includegraphics[width = 0.6\textwidth]{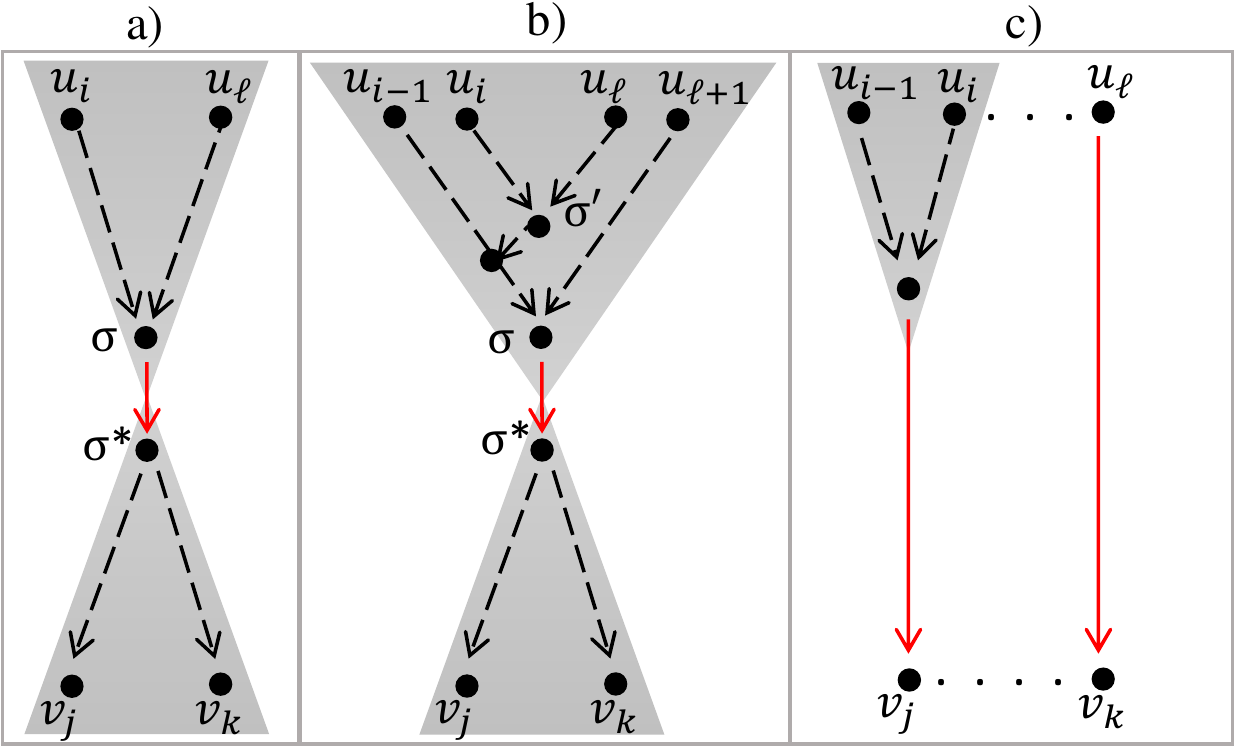}
    \caption{Schematic representation of the three cases for proving Lemma \ref{lemma_deviation_epsilon_block}, depending on the structure of a block $\mathcal{B}$. a) Case 1.1: $N_{\mathcal{T}}(x) = N_{\mathcal{B}}(x)$. b) Case 1.2: $\vert N_{\mathcal{T}}(x)\vert > \vert N_{\mathcal{B}}(x)\vert$. c) Case 2: The nodes in $N_{\mathcal{B}}(x)$ belong to different trees.}
    \label{fig:block_trees}
\end{figure}

\begin{proof} We distinguish whether all nodes of a block  $\mathcal{B}$ belong to the same tree or if they are in different trees. 
Without loss of generality, we specify monotonic paths and trees to be increasing. \\
\indent\textit{Case 1:} All nodes $N_{\mathcal{B}}(x)$ are in one tree $\mathcal{T}$  (see Figure \ref{fig:block_trees}a,b). Without loss of generality, $\mathcal{T}$ is considered to be a tree of Type~4, since all other cases follows the same or a simpler argumentation.
We further distinguish whether the nodes of $N_{\mathcal{B}}(x)$ are the only nodes in $\mathcal{T}$.\\
\indent\textit{Case 1.1:} $N_{\mathcal{T}}(x) = N_{\mathcal{B}}(x)$. 
For the bottleneck nodes it holds that $\val(\sigma) = x_i$ and $\val(\sigma^*)\in V(y)$ (see Corollary \ref{lemma:sigma}). 
The replacement of $x$ by $x_{\varepsilon, i, \ell}$ is a beneficial increase for all $\varepsilon \in [0,\varepsilon']$ with $\varepsilon' = \sigma^*-x_i$ because the intermediate node is also shifted to $\sigma + \varepsilon$ that leads to lower move cost of $\val(\sigma^*) - \val(\sigma) - \varepsilon$, that is, a decrease by $\varepsilon$. 
Therefore, we get the left side of the first equation $ d(x_{\varepsilon, i,\ell}, y) + \varepsilon = d(x,y)$ for~$\rho_I=1$.
It holds that $x_i + \varepsilon' = \val(\sigma^*) \in  V(y) $.   
For the right side of the equation with $\rho_D=1$, the argumentation is similar. Replacing  $x$ by $x_{-\varepsilon, i,\ell}$ we get new move cost of  $\val(\sigma^*) - \val(\sigma) + \varepsilon$, that is, an increase by $\varepsilon$.

\indent\textit{Case 1.2:} $\vert N_{\mathcal{T}}(x)\vert > \vert N_{\mathcal{B}}(x)\vert$. 
Since $\mathcal{T}$ is a tree of Type~4, all nodes in $N_{\mathcal{T}}(x)$ merge to the intermediate node $\sigma$. 
Moreover, it is evident that the merge of adjacent points in $\mathcal{T}$ that are equal creates lower cost than merging two points that are different. 
Therefore, there exists an intermediate node $\sigma'$ where all nodes in $N_{\mathcal{B}}(x)$ merge to (see Figure \ref{fig:block_trees}b). 
Then Lemma \ref{lemma:deviation_epsilon} can be applied for $u_i =  \sigma'$ with  $\rho_I= \rho_D =1$.
Depending on the case, an $\varepsilon_I$ is specified such that we get one of the above equations for an $\varepsilon \in [0,\varepsilon_I]$. 
For $\varepsilon = \varepsilon_I$, the block is shifted until it reaches a value of the adjacent points of the block, that is $ x_{i-1}$ or $x_{\ell +1}$, or it reaches a point in $V(y)$. 
\\

\indent\textit{Case 2:} The nodes in $N_{\mathcal{B}}(x)$ belong to different trees (see Figure \ref{fig:block_trees}c).
In this case, we need to count how many trees we have beneficial increases and decreases.
To decide whether a replacement is beneficial or detrimental, there are two possible cases of trees belonging to the block $\mathcal{B}$. 
The first case is that all nodes in a tree in $\mathcal{B}$ belong to $N_{\mathcal{B}}(x)$. 
Then we can apply Case 1.1. 
The second case concerns the boundary values of $N_{\mathcal{B}}(x)$ merging with the predecessors or successors of the block~$\mathcal{B}$. 
Following the argumentation of Case~1.2, we determine if the replacement of~$x$ by~$x_{\varepsilon, i,\ell}$ is beneficial, neutral, or detrimental.  
Ignoring neutral replacements, we set~$x_{i,\ell}^+$ as the number of trees for which we have a beneficial increase and~$x_{i,\ell}^-$ as the number of trees for which  we have a beneficial decrease. 
Shifting a whole block may therefore lead to a reduction of distance of more than $\varepsilon$. We get the above statement for $\rho_I =x_{i,\ell}^+$  and $\rho_D =x_{i,\ell}^- $.
By Lemma \ref{lemma:deviation_epsilon}, we get an $\varepsilon_{\mathcal{T}}$ for all trees $\mathcal{T}$ in a block $\mathcal{B}$ restricting a beneficial or neutral replacement. Let $\varepsilon^+_{min}$ be the minimum of all   $\varepsilon_{\mathcal{T}}$ for which  we have a beneficial increase. Analogously,  $\varepsilon^-_{min}$ is defined for beneficial decreases. 
Without loss of generality, we assume $x_{i,\ell}^+ > x_{i,\ell}^- $. 
Hence, it holds that for $x_i + \varepsilon^+_{min}$ is in $\{x_{i-1}, x_{\ell+1}\}$ or in $V(y)$. 
\end{proof}

\subsection{MSM-Mean Values}

We now use beneficial and neutral replacements to prove that for any set $X$ there exists a mean $m$ such that all points of $m$ are points of at least one time series of $X$.   

\begin{lemma}
\label{lemma:values_of_m}
Let $X = \{x^{(1)}, \ldots, x^{(k)}\}$ be a set of $k$ time series. Then there exists a mean $m~=~(m_1,\ldots,m_N)$ of $X$ such that $m_i \in V(X)$ for all $m_i, i\in [N]$.
\end{lemma}

\begin{proof}
Assume towards a contradiction that every mean has at least one point that is not in $V(X)$. Among all means, choose a mean~$m$ such that 1)  $n_V$, the number of points of~$m$ that are in $V(X)$, is maximum and 2) among all means with $n_V$ points from~$V(X)$, the number of transitions from $m^{(i)}_j$ to $m^{(i)}_{j+1}$ where  $m^{(i)}_j \neq m^{(i)}_{j+1}$ is minimum.  
In other words,~$m$ has a minimal number of blocks.
Let~$\mathcal{B}$ be a block in~$m$ whose points are not in~$V(X)$.
We apply Lemma \ref{lemma_deviation_epsilon_block} to show that there exists an $\varepsilon\in \mathds{R}$ such that  $m_{\varepsilon, i,\ell}$ is a mean where the points of the shifted block $\mathcal{B}$ reach a point of a predecessor or successor of $\mathcal{B}$ or a point in $V(X)$.
We now specify $\varepsilon$. 
First, we determine whether $\varepsilon$ is positive or negative. 
For each sequence, one of the Cases (1) to (3) of Lemma \ref{lemma_deviation_epsilon_block} applies. 
For each time series in $X$, we introduce two variables to count how many beneficial increases and decreases we have. 
Neutral replacements are not counted.  
Let $x^+$ be the sum of~$\rho_I$ for beneficial increases and  $x^-$ be the sum of~$\rho_D$ for beneficial decreases of all time series. 
If $x^+ \geq x^-$, we set $\varepsilon$ as the minimum of the specified $\varepsilon'$ for beneficial increases and all $\varepsilon_I$ (see Lemma~\ref{lemma_deviation_epsilon_block}).
If $x^+ < x^-$ we set $\varepsilon$ as the maximum of the specified $-\varepsilon'$ for beneficial decreases and all $-\varepsilon_D$.
Compared to the mean $m$, all values of $m_{\varepsilon, i,\ell}$ are the same except the values of the shifted block.
By Lemma \ref{lemma_deviation_epsilon_block} the points of the new mean $m_{\varepsilon, i,\ell}$ are shifted for the specified~$\varepsilon$  until they reach a point of the right or left neighbor block or a point in $V(X)$.
If they reach a point of the right or left neighbor block, we have a contradiction to the selection of a mean with a minimal number of transitions.
If they reach a point in $V(X)$, we have the contradiction to the selection of a mean with minimal number of values that are not in $V(X)$.  
\end{proof}

\section{Computing an exact MSM-Mean}
\label{chap:computing_MSM_mean}
Based on Lemma \ref{lemma:values_of_m}, we now give a dynamic program computing a mean $m$ of~$k$ time series $X=\{x^{(1)},\ldots,x^{(k)}\}$. 
The transformation operations are described for the direction transforming $X$ to $m$. 
\subsection{Dynamic Program}
We fill a $(k+2)$-dimensional table $D$ with entries $D[(p_1, \ldots, p_k),\ell, s]$, where \begin{itemize}
    \item $p_i \in [n_i]$ indicates the \emph{current position} in time series $x^{(i)}$,
    \item the index $\ell\in [N]$ indicates the \emph{current position of $m$}, and
    \item $s$ is the index of a point $v_s \in V(X)$.
\end{itemize}   
We also say that $(p_1, \ldots, p_k)$ are the \emph{current positions of $X$}. 
For clarity, we write $p= (p_1,\ldots, p_k)$.
The entry $D[p,\ell, s]$ represents the cost of the partial time series $\{(x_1^{(1)}, \ldots, x_{p_1}^{(1)}),\ldots, (x_{1}^{(k)}, \ldots, x_{p_k}^{(k)})\}$ transforming to a mean $(m_1,\ldots,m_{\ell})$ assuming that $m_{\ell} = v_s$.
Giving a recursive formula filling table $D$ we have two transformation cases. 
The case distinction is based on the computation of the \textsc{MSM} distance for two time series $x = (x_1,\ldots, x_m)$ and $(y_1,\ldots, y_n)$.
This computation fills a two-dimensional table $D^*$. 
An entry $D^*[i,j]$ represents the cost of transforming the partial time series $(x_1, \ldots, x_i)$ to the partial time series $(y_1,\ldots, y_j)$. The distance $d(x,y)$ is given by $D^*[m,n]$.
Stefan et al. \cite{stefan2012move} give the recursive formulation of the \textsc{MSM} metric as the minimum of the cost for the three transformation operations.
\begin{align*}
    D^*[i,j] &= \min\{A_{MO}[i,j], A_{M}[i,j], A_{SP}[i,j] \}, \text{where} &\\
    A_{MO}[i,j] &= D^*[i-1,j-1] + \vert x_i -y_i \vert & (move) \\
    A_{M}[i,j] &= D^*[i-1,j] + C(x_i, x_{i-1}, y_j) & (merge) \\
    A_{SP}[i,j] &= D^*[i,j-1] + C(y_j, x_i, y_{j-1}) & (split) 
\end{align*}
\noindent for
$$
C(x_i,x_{i-1},y_j) = \begin{cases} 
c \;\;\;\; \text{if } x_{i-1} \leq x_i \leq y_j \text{ or } x_{i-1}\geq x_i \geq y_j &  \\
c + \min(\vert x_i-x_{i-1}\vert,\vert x_i-y_j\vert) \quad \text{otherwise}. & 
\end{cases}
$$
When the recursion reaches a border of $D^*$, we have the special cases $D^*[i,1] = D^*[i-1,1] + C(x_i, x_{i-1}, y_1) $ (only merge operation may be further applied) and $D^*[1,j] =D^*[1,j-1] + C(y_j, x_1, y_{j-1})$ (only split operation may be further applied). 
The base case is reached for $D^*[1,1] = \vert x_1 -y_1\vert$ where only a move operation is applied. 

For the recursion formula for the \textsc{MSM-Mean} problem, we distinguish between applying moves and splits ($A_{MS}$) and  only merges ($A_{ME}$):
\begin{align*}
    D[p,\ell, s] = \min \{ A_{MS}[p,\ell, s], A_{ME}[p,\ell, s]\}.
\end{align*}
To distinguish between these cases, we introduce index sets $I_{MO}, I_{SP}$, and $I_{ME}$ for move, split, and merge operations, respectively. 
They represent the indices for those time series which either move, split, or merge.  
All index sets are subsets of $I=[k]$.
Let $\overline{p}_{MO}$ be the tuple of $p$ where $\overline{p}_i = p_i-1$ for all $i\in I_{MO}$. 
The tuple $\overline{p}_{ME}$ is defined analogously. 
The first case considers that at some current positions of $X$ there are move and at all other positions there are split operations. 
It holds that $I_{MO} \cup I_{SP}=I$.
For these operations, the recursive call of the function decreases the current position of~$m$:
\begin{align*}
A_{MS}[p,\ell, s] = &\min_{v_{s'}\in V(X)}\{\min_{I_{MO},I_{SP}} \bigl( D[\overline{p}_{MO}, \ell-1, s'] \\
&+ \sum_{i\in I_{MO}} \vert x_{p_i}^{(i)} - v_s\vert + \sum_{i\in I_{SP}} C(v_s,x_{p_i}^{(i)}, v_{s'})\bigr)\}. 
\end{align*}
The second case treats merge operations of at least one current position of $X$.
If a merge is applied, all other time series pause since the recursive call does not decrease the position of $m$:
\begin{align*}
A_{ME}[p,\ell, s] = &\min_{I_{ME}} \bigl(D[\overline{p}_{ME},\ell, s] +  \sum_{i\in I_{ME}} C(x_{p_i}^{(i)},x_{p_i-1}^{(i)}, v_{s}) \}. 
\end{align*}

\noindent For the last step in the recursion, the entries $D[(1,\ldots, 1), 1, s]$ for all $v_s \in V(X)$ are calculated by
\begin{align*}
    D[(1,\ldots, 1), 1, s] = \sum_{i\in I} \vert x_{1}^{(i)} - v_s\vert.
\end{align*}
All entries $D[p,\ell, s]$ for which  $p_i<1$ for some $i\in [k]$ is set to $+\infty$.
If $\ell =1$ and all $p_i>1$, only merge operations may be applied:
\begin{align*}
  D[p,1, s] = \min_{I_{ME}\subseteq I} \bigl(D[\overline{p}_{ME},1, s] +  \sum_{i\in I_{ME}}C(x_{p_i}^{(i)},x_{\overline{p}_i}^{(i)}, v_{s}). 
\end{align*}

The correctness of the dynamic program hinges on the fact that in the recursive definition of the pairwise distance, the value of $D^*[i,j]$ depends only on the values of $D^*[i,j-1]$, $D^*[i-1,j]$, $x_i,y_j,x_{i-1}$, and $y_{j-1}$; we omit the formal correctness proof.

\subsection{Running Time Bound}
We now show an upper bound on the maximum mean length in terms of the total length of $X$. 
To this end, we first make the observation that the index set  $I_{MO}$ is never empty.
That is, it is not optimal to apply only split operations in one recursion step. 

\begin{lemma}
\label{lemma:mergeAndMove}
Let $m$ be a mean of a set of $k$ time series $X$. It holds that  
$
D[p,\ell, s] < \min_{v_{s'}\in V(X)} \{ D[p,\ell-1, s'] + \sum_{i\in I} C(v_s,x_{p_i}^{(i)}, v_{s'})\}
$.
\end{lemma}

\begin{proof}
Let $D(X,m)$ be the distance of a mean $m$ to $X$. 
Assume towards a contradiction, that there exists a recursion step, where $I_{MO}= \emptyset$. 
That is, in each time series in $X$ there is a split at a point $x_{p_i}^{(i)}, i\in I$ to the points $m_{\ell}$ and $m_{\ell -1}$. 
We regard the cost for the transformation up to the positions $(p_1, \ldots, p_k)$ of~$X$ and $\ell$ of $m$. 
Applying the recursion formula for $I_{SP}=I$, we get $D[p,\ell, s] =\min_{v_{s'}\in V(X)} \{D[p,\ell-1, s'] + \sum_{i\in I} C(v_s,x_{p_i}^{(i)}, v_{s'})\}.$ 
Let $m'$ be a mean of $X$ equal to $m$ but where $m_{\ell}$ is deleted. 
For the mean $m'$, we save the cost for splitting $\sum_{i\in I} C(v_s,x_{p_i}^{(i)}, v_{s'})$, without changing the alignment of all other points in $X$. It follows that $D(X,m') < D(X,m)$.
This is a contradiction to $m$ being a mean. 
\end{proof}

Lemma~\ref{lemma:mergeAndMove} now leads to the following upper bound for the \textsc{MSM-Mean} length. 

\begin{lemma}
\label{lemma:mean_length}
Let $X=\{x^{(1)}, \ldots, x^{(k)}\}$ be a set of time series with maximum length $max_{j\in[k]}\vert x^{(j)} \vert = n$. Then, every mean~$m$  has length at most $(n-1)k + 1$.
\end{lemma}

\begin{proof}
Towards a contradiction, let $m$ be a mean of $X$ with length $N>(n-1)k + 1$.
The entry of the first recursion call is $D[(n_1, \ldots, n_k), N, s]$.
Consider any sequence of recursion steps from $D[(n_1, \ldots, n_k), N, s]$ to $D[(1, \ldots, 1), \cdot, \cdot]$; each step is associated with  index sets $I_{MO}$ and~$I_{SP}$, or $I_{ME}$.
By Lemma \ref{lemma:mergeAndMove}, it holds that $I_{MO}~\neq~\emptyset$ in each step. 
That is, at least one current position of $X$ is reduced by one in each recursion step until the entry $D[(1,\ldots, 1), \ell', s']$ is reached. 
These are at most $(n-1)k+1$ recursion steps. 
Since $N~>~(n-1)k + 1$, it holds that $\ell'>1$. 
The only possibility for a further recursion step for $D[(1,\ldots, 1), \ell', s']$ is to set $I_{SP}= I$, since~$D[p,\ell, s]=+\infty $ whenever $p_i<1$ for some~$i$.
By Lemma \ref{lemma:mergeAndMove}, we get a contradiction to $m$ being a mean.  
\end{proof}

We now bound the running time of our algorithm.
\begin{lemma}
The \textsc{MSM-Mean} problem for $k$ input time series of length at most~$n$ can be solved in time $ \mathcal{O}(n^{k+3} 2^k k^3)$. 
\end{lemma}

\begin{proof}
In the dynamic programming table $D$ at most $n^{k+2}k^2$ entries have to be computed. 
This number is the dimension of $k$ time series with maximum length $n$, the maximum length of the mean $(n-1)k +1 \leq kn$ and the size of $V(X)$ which is at most $kn$. 
For each table entry, the minimum over the set $V(X)$ is taken, which includes again $kn$ data points. For each minimum over $V(X)$ all subsets of $[k]$ are considered which are at most $2^k$ sets. 
All subsets of $[k]$ are only generated once for both $I_{MO}$ and $I_{ME}$.
Thus, filling the table iteratively takes time $\mathcal{O}(n^{k+3} 2^k k^3)$. 

For the traceback, the start entry of $D$ is any of the position $(n_1,\ldots,n_k)$ with minimal cost, that is, 
$$D[(n_1,\ldots,n_k),\ell_{start},s_{start}]  = \min_{\ell, s} D[(n_1,\ldots,n_k),\ell,s].
$$ The length of the mean $m$ is $\ell_{start}$ with $m_{start} = v_{s_{start}}$. 
In each traceback step, the \emph{predecessors} of the current entry are determined, that are the entries leading to the cost of the current entry. A predecessor of an entry is not unique. 
For setting the mean data point we consider the current entry $D[(p_1,\ldots,p_k),\ell,s]$ and the entry of the predecessor  $D[(q_1,\ldots,q_k),\ell',s']$.
If $\ell' = \ell-1$, the point $ v_{s'}$ is assigned to the mean point $m_{\ell'}$ and we continue with the next traceback step. Otherwise, the next traceback step is directly applied without assigning a mean point. 
We repeat this procedure until we reach the entry $D[(1,\ldots,1),1,s^*]$. 
The running time of filling the table clearly dominates the linear time for the traceback.
 \end{proof}
 
 \subsection{Implementation \& Window Heuristic}
 \label{subsec:Implementation_Heuristic}
We fill the table $D$ iteratively and apply the above described traceback mechanism afterwards.
Since the running time of \textsc{MSM-Mean} will often be too high for moderate problem sizes, we introduce the \emph{window heuristic} to avoid computing all entries of table $D$.
Similar to a heuristic of Levenshtein distance \cite{ukkonen1985algorithms}, the key idea is to introduce a parameter $d$ called the window size representing the maximum difference between the current positions of the time series within the recursion. 
All entries whose current positions are not within distance $d$ will be discarded. 
For example, an entry with current position $(6,3,4)$ of $X$ will not be computed for $d=2$. 
In the case of a set of time series with unequal lengths, where $n_{min}$ and $n_{max}$ denotes the minimum length and the maximum lengths, respectively, of all time series, $d$ has to be greater than $n_{max} - n_{min}$. 
 
\section{Experimental Evaluation}
\label{chap:experimental_eval}
This section provides important results from a selection of experiments using implementations of mean algorithms\footnote{All code is available on GitHub: \url{https://github.com/JanaHolznigenkemper/msm_mean}}. After a description of the experimental setup, we first provide a running time comparison of the \textsc{DTW-Mean} algorithm \cite{brill2019exact} and our \textsc{MSM-Mean} algorithm. 
Furthermore, we examine accuracy and running times of \textsc{MSM-Mean} for various heuristics. 

\subsection{Experimental Setup}
The running times of our Java implementations are measured on a server with Ubuntu Linux 20.04 LTS, two AMD EPYC 7742 CPUs at 2.25 Ghz (2.8 Ghz boost), 1TB of RAM, Java version 15.0.2. Our implementations are single-threaded. For our results at most 26GB of RAM were occupied.  

The experiments are conducted on 20 UCR data sets \cite{UCRArchive} that Stefan et al. \cite{stefan2012move} already used (see Table \ref{tab:testsetsUCR}).
The UCR data sets were collected for the use of time series classification. Each set consists of a training and a testing set. The data sets consists of time series of different classes and different lengths.
Since we are not using the sets for classification use cases yet, we just take the training sets of each data set for our experimental setup.
The parameter $c$ is set constant for every data set following the suggestions of Stefan et al. \cite{stefan2012move}.

\begin{table}[t]
 \caption{List of 21 UCR time series data sets. For our running time experiment, we did not take the Italy Power Demand data set (*) since the time series are too short. The quality analysis of \textsc{MSM-Mean} using this data set was conducted for $c\in\{0.01,0.1,0.2,0.5\}$. }
    \centering
    \begin{tabular}{l|c|c|c|c}
         Data Set& \#Classes & \#TS Training & TS Length & MSM $c$  \\ \hline
         50words&50&450&270&1\\
         Adiac &37&390&176& 1 \\
         Beef &5&30&470&0.1 \\
         CBF &3& 30 &128&0.1 \\
         Coffee &2&28&286&0.01 \\
         ECG &2&100&96&1 \\
         FaceAll &14&560&131&1 \\
         Face (four) &4&24&350&1\\
         Fish &7&175&463&0.1 \\
         Gun Point &2& 50&150& 0.01 \\
         Italy Power* &2&67&24&*\\
         Lightning-2 &2&60&367&0.01\\
         Lightning-7 &7&70&319&1\\
         OliveOil &4&30&470&0.01\\
         OSU Leaf &6&200&427&0.1\\
         Swedish Leaf &15&500&128&1 \\
         Synthetic C.  & 6& 300& 60& 0.1 \\
         Trace &4&100&275&0.01\\
         Two Pattern &4&1000&128&1\\
         Wafer  &2&1000&152&1\\
         Yoga &2&300&426&0.1 
         \end{tabular}
    \label{tab:testsetsUCR}
\end{table}

Due to the complexity of the algorithms, we draw time series samples from the training sets obtained from the UCR archive in the following way. For each class of the training sets, we randomly pick $k$ time series, $k\in  \{3,4,5\}$, and for each of them, we cut out a contiguous subsequence of length $n$ starting at a random data point, $n \in \{10, \ldots,50\}$. In addition, we limit the length of the mean time series to at most $n$.

\subsection{Running Time Comparisons}

In our first experiment, we consider the running times for $k=3,4$. Figure~\ref{fig:lineplotMSMDTW} shows the average running time over all 20 data sets as a function of $n$. Our  \textsc{MSM-Mean} algorithm is substantially faster than the DTW counterpart.
The outlier in the DTW graph is due to a data set where the implementation does not complete within 
10 minutes. 
Figure \ref{fig:boxplotMSMDTW} in the appendix provides box plots depicting the running times of both algorithms for $k=4$ and $n=10 \ldots, 13$. They reveal that the \textsc{MSM-Mean} algorithm has smaller medians and interquartile ranges and fewer outliers of high running times  compared to the \textsc{DTW-Mean} algorithm.  
The \textsc{MSM-Mean} implementation was able to compute any instance for $k=3$, $n < 43$, $k=4$, $n < 19$, and $k=5$, $n < 11 $ within 10 minutes.
For \textsc{DTW-Mean}, this was only the case for $k=3$, $n<29$ and $k=4$, $n<14$.

\begin{figure}[h]
    \centering
    \includegraphics[width = 0.45\textwidth]{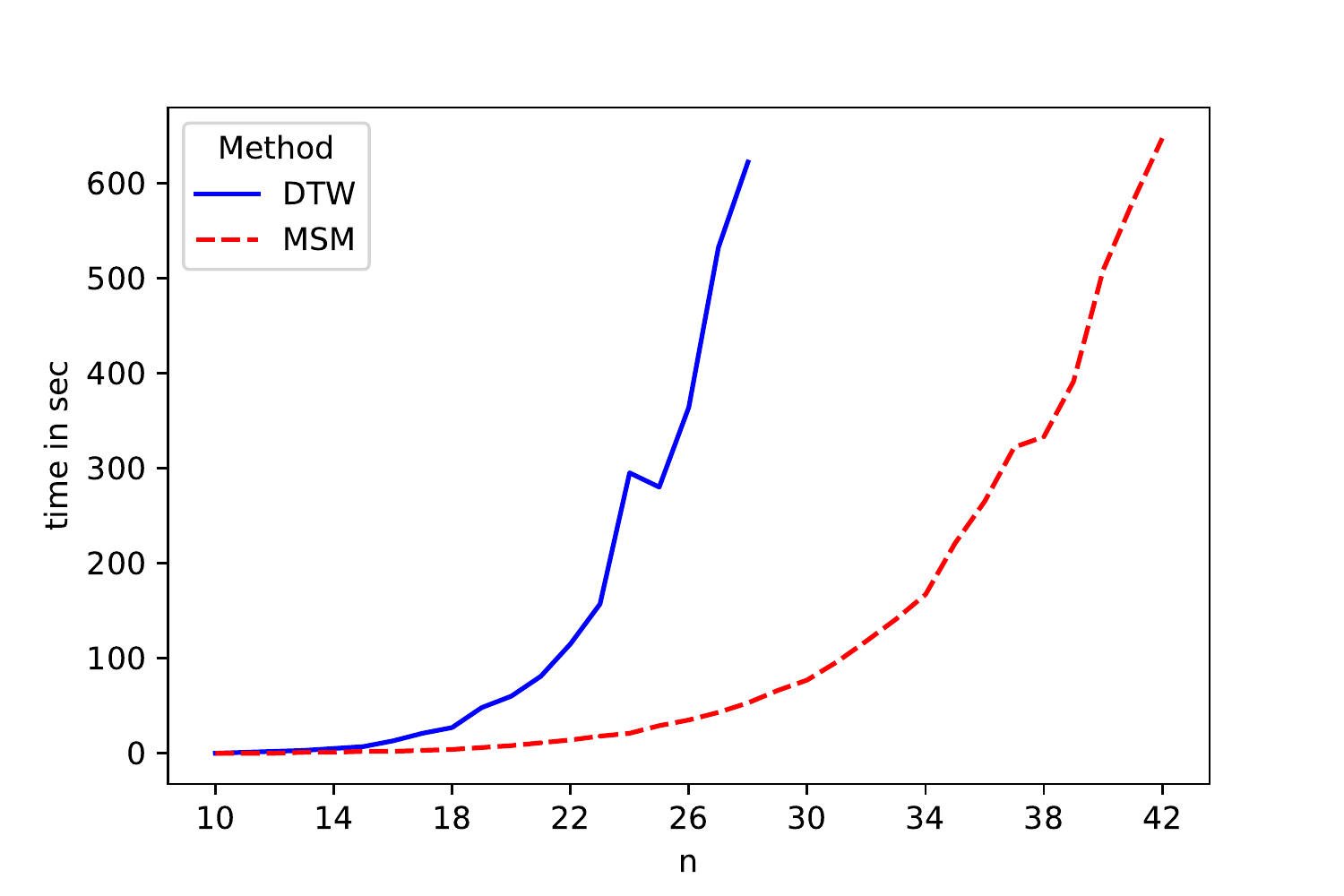}
    \caption{Running time comparison of \textsc{MSM-Mean} and \textsc{DTW-mean} for $k=3$ as a function of $n$}
    \label{fig:lineplotMSMDTW}
\end{figure}

\begin{figure}[h]
    \centering
    \includegraphics[width = 0.90\textwidth]{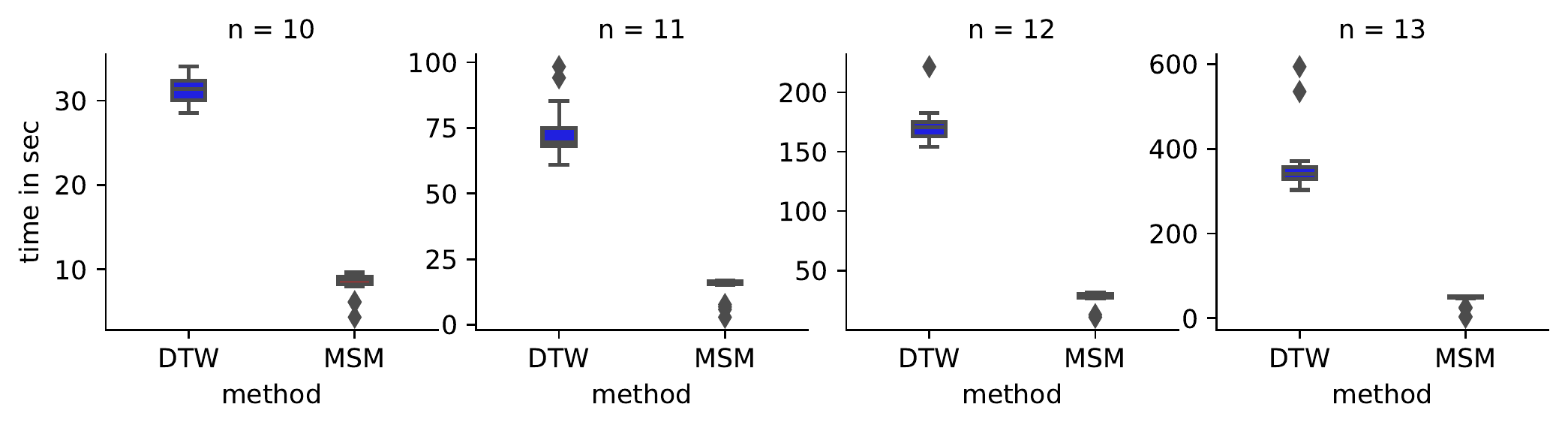}
    \caption{Running time comparison of \textsc{MSM-Mean} and \textsc{DTW-mean} for $k=4$}
    \label{fig:boxplotMSMDTW}
\end{figure}

\subsubsection{\textsc{MSM-Mean} Quality}
To evaluate the quality of the \textsc{MSM-Mean}, we use the algorithm on the \textit{ItalyPowerDemand} data set \cite{UCRArchive} where each time series has length 24. The data set contains two classes. 
For different values of $c \in \{0.01,0.1,0.2,0.5\}$, Table \ref{tab:Italy} shows the distance of \textsc{MSM-Mean} to three other time series. The first row reports the distance when the time series belong to one class, while the second row provides the distance when taking them from both classes. The results confirm for all $c$ that distances of \textsc{MSM-Mean} are lower when time series belong to one class. 

\begin{table}[h]
    \centering
     \caption{Distance of the mean to time series of the data set \textit{ItalyPowerDemand} of one class and to time series of mixed classes for k=3 and n=24 for varying $c$}
    \begin{tabular}{|l|c|c|c|c|}
        \hline
         \textbf{c} & \textbf{0.01} & \textbf{0.1} & \textbf{0.2} & \textbf{0.5}  \\ \hline
         \textbf{one class} & 5.55 & 7.05 & 8.02 & 9.72 \\ \hline
         \textbf{two classes} & 11.66 & 15.41 & 18.46 & 26.1 \\  \hline
    \end{tabular}
    \label{tab:Italy}
\end{table}

\subsubsection{Length of the Mean} 
We implemented two versions of the \textsc{MSM-Mean} algorithm, one with a fixed length $n$ of the mean and one without length constraints. As shown in Lemma~\ref{lemma:mean_length}, the lengths of \textsc{MSM-Mean} is at most $(n-1)k+1$. However, the results of our experiments for $k=3$ and $n\in\{10,\ldots,30\}$ show that the length of \textsc{MSM-Mean} is always exactly $n$. Thus, it is advisable to use this constraint as done in the experiments discussed above.

\subsection{\textsc{MSM-Mean} Heuristics}

\subsubsection{Discretization Heuristic}
Because the domain size of the values of a time series has an significant effect on the performance of \textsc{MSM-Mean} algorithm, we propose a second heuristic where the domain is split into $v$ buckets of equal length. Each value $x$ of a time series is then replaced by the center point of the bucket to which $x$ belongs. Thus, there are at most $v$ different values in total. 

Figure \ref{fig:maxV} shows the running time of this heuristic as a function of $v$ for $k=3$ and $n=30$. There is a substantial (close to linear) decrease in the running time with a decreasing number of buckets. Moreover, the relative error is quite moderate: We observed an average and maximum error of 4.6\% and 8.47\%, respectively. 

\begin{figure}[t]
    \centering
    \includegraphics[width = 0.5\textwidth]{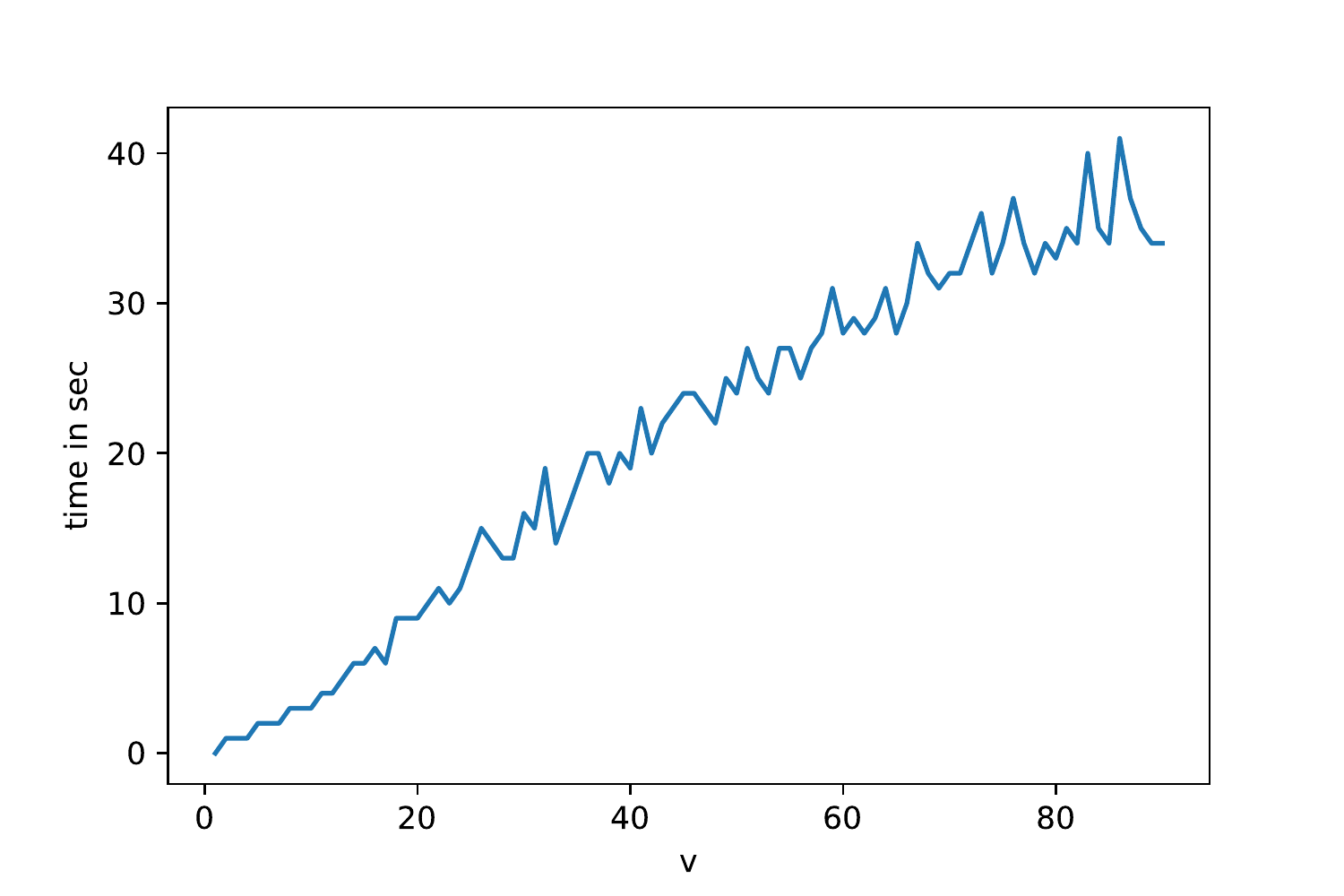}
    \caption{Running time for \textsc{MSM-Mean} calculation of $k=3$ and $n=30$ regarding the discretization heuristic.}
    \label{fig:maxV}
\end{figure}

\subsubsection{Window Heuristic}
In the following, we investigate the window heuristic described in Section \ref{subsec:Implementation_Heuristic} for the \textsc{MSM-Mean} problem and $k=3$, $n\in \{10,\ldots, 42\}$. We examine the window size $d=1,2,3$ in our experiments. 
We analyze relative error of the exact mean and the means obtained from the window heuristic.
As expected, the higher the window size $d$ the smaller is the relative error. 
The relative error averaged over all $n$ and all data sets was $4.8\%, 3.2\%,2.4\%$ and the maximum relative error was $9.1\%,6.4\%,5.4\%$ for $d=1,2,3$, respectively.  
Figure \ref{fig:Heuristic_runningtime} shows the running time of the window heuristic in comparison to the exact computation as a function of $n$. Note that the y-axis plots the running time on a logarithmic scale. 
For all parameter settings, the running times improve substantially in comparison to the exact approach.

\begin{figure}[h]
    \centering
    \includegraphics[width = 0.5\textwidth]{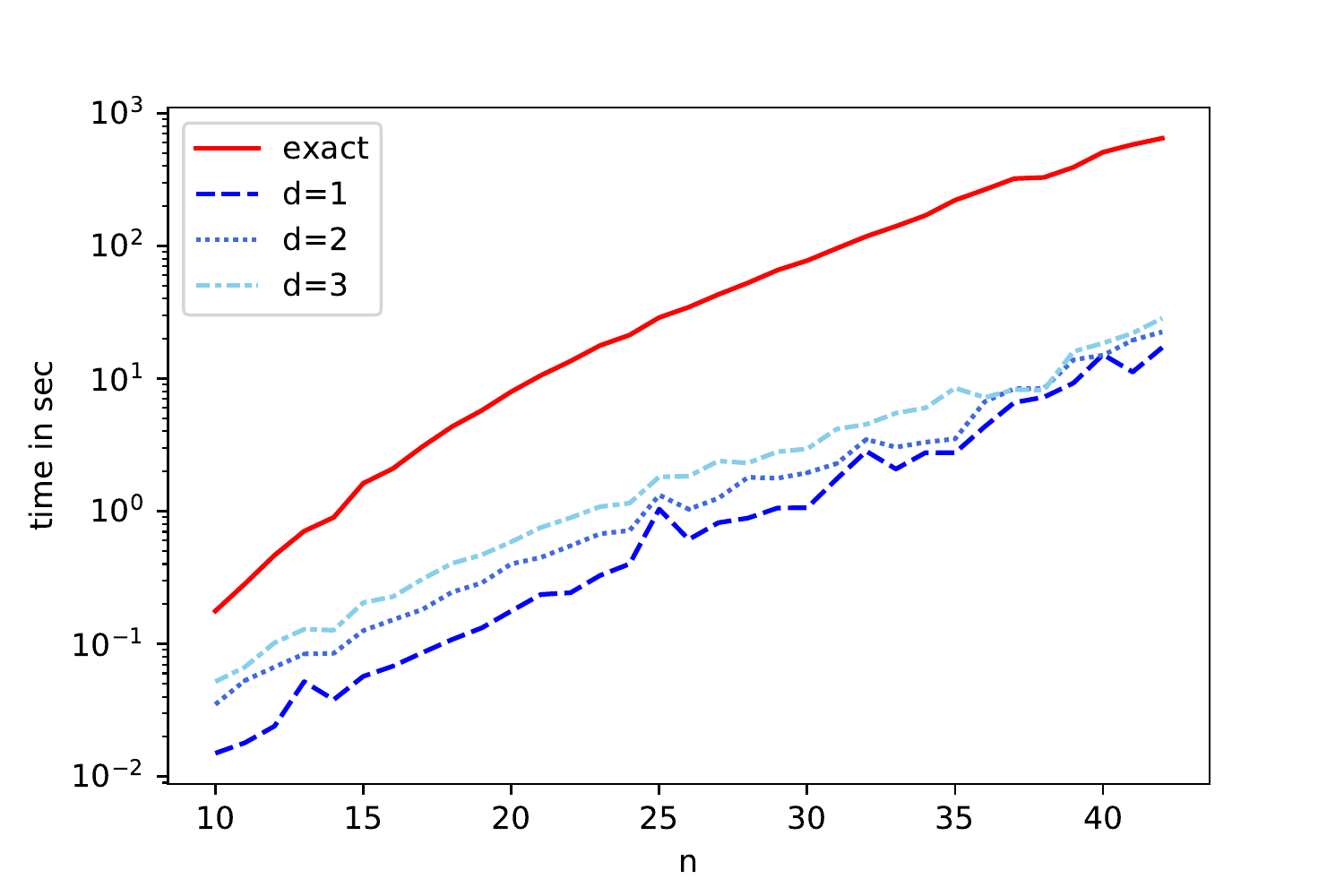}
    \caption{Average running time of computing the exact mean and mean using the window heuristic  for $d=1,2,3$ for $k=3$ and $n\in \{10,\ldots, 42\}$.}
    \label{fig:Heuristic_runningtime}
\end{figure}

\section{Conclusion and Future Work}
\label{chap:conclusion}
This paper introduces the \textsc{MSM-Mean} problem of computing the mean of a set of time series for the Move-Split-Merge (MSM) metric. We present an exact algorithm for \textsc{MSM-Mean} with a better running time than a recent algorithm for computing the mean for the DTW distance. Experimental results confirm the theoretically proven superiority of our \textsc{MSM-Mean} algorithm in comparison to the DTW counterpart. The key observation of our method is that an \textsc{MSM-Mean} exists whose data points occur in at least one of the underlying time series. In addition, we provide an an upper bound for the length of \textsc{MSM-Mean}. In our experiments, the maximum mean length is much shorter, rarely exceeding the length of the longest time series. The paper also provides two heuristics for speeding up the computation of the mean without sacrificing much accuracy, as shown in our experimental evaluation.

In future work we will tackle the following issues. First, we will examine how to use \textsc{MSM-Mean} in real clustering and classification problems. Second, we plan to develop optimization strategies such as the \emph{A*-Algorithm} \cite{hart1968formal} for further improving the running time of our algorithm to avoid filling up the entire dynamic programming table. As a starting point, the structure of the transformation forests and the metric properties of the \textsc{MSM} distance could be further explored. 
Third, the metric properties, especially the triangle inequality, of the \textsc{MSM} distance enables applications of the \textsc{MSM-Mean} to metric indexing.  
Finally, we conjecture \textsc{MSM-Mean} to be NP-hard.
Proving this conjecture could be a next research step. 

\section*{Declarations}
The authors have no relevant financial or non-financial interests to disclose.


\printbibliography

\end{document}